\documentclass[11pt]{article}

\usepackage{amsmath, amssymb, amsthm, comment, graphicx, lineno, color}
\usepackage{algorithm}
\usepackage{adjustbox}
\usepackage{caption}    
\usepackage{setspace}
\usepackage{subcaption}
\usepackage{placeins}
\usepackage{float}
\usepackage[utf8]{inputenc}
\usepackage{multirow}
\usepackage{mathrsfs}
\usepackage{fullpage}

\usepackage[figuresright]{rotating}
\usepackage{tikz, subcaption, pgfplots}
\usetikzlibrary{shapes, arrows, positioning}
\usepackage{enumitem}
\usepackage[round]{natbib}
\usepackage{url} 
\usepackage{siunitx}
\usepackage[colorlinks, citecolor=blue]{hyperref}

\theoremstyle{definition}

\newtheorem{assumption}{Assumption}

\theoremstyle{plain}
\newtheorem{theorem}{Theorem}
\newtheorem{corollary}[theorem]{Corollary}
\newtheorem{lemma}[theorem]{Lemma}
\newtheorem{proposition}[theorem]{Proposition}

 \usepackage[affil-it]{authblk}

\DeclareMathOperator{\Var}{Var}
\definecolor{mycolor2}{RGB}{218,165,32}

 \numberwithin{equation}{section}

\newcommand{\init}{\mbox{{\tiny init}}}
\newcommand{\corr}{\mbox{{\tiny corr}}}
\newcommand{\vertiii}[1]{{\left\vert\kern-0.25ex\left\vert\kern-0.25ex\left\vert #1 
    \right\vert\kern-0.25ex\right\vert\kern-0.25ex\right\vert}}

\title{Statistical significance in high-dimensional linear mixed models}

\author[a]{Lina Lin}
\author[b]{Mathias Drton}
\author[c]{Ali Shojaie}
\affil[a]{Department of Statistics, University of Washington}
\affil[b]{Department of Mathematics, Technical University of Munich}
\affil[c]{Department of Biostatistics, University of Washington}
\date{\today}

\begin{document}
\maketitle

\begin{abstract}
  This paper concerns the development of an inferential framework for
  high-dimensional linear mixed effect models. These are suitable
  models, for instance, when we have $n$ repeated measurements for $M$
  subjects. We consider a scenario where the number of fixed effects
  $p$ is large (and may be larger than $M$), but the number of random
  effects $q$ is small. Our framework is
  inspired by a recent line of work that proposes de-biasing penalized
  estimators to perform inference for high-dimensional linear models
  with fixed effects only. In particular, we demonstrate how to
  correct a `naive' ridge estimator in extension of work by
  \cite{buhlmann2013} to build asymptotically valid confidence
  intervals for mixed effect models. We validate our theoretical
  results with numerical experiments, in which we show our method
  outperforms those that fail to account for correlation induced by
  the random effects.  For  a practical
  demonstration we consider a riboflavin production dataset that exhibits group
  structure, and show that conclusions drawn using our method are
  consistent with those obtained on a similar dataset without group
  structure.
\end{abstract}

\section{Introduction}
Modern statistical problems are increasingly high-dimensional, with
the number of covariates $p$ potentially vastly exceeding the sample
size $N$. This is due in part to technological advances that
facilitate data collection. For instance, we are now able to measure the expression of many genes in a given specimen at little cost. However, it often remains expensive to have many replicates/species to experiment on, resulting in $N \ll p$. 

Fortunately, significant progress has been made in developing rigorous statistical tools for tackling such problems. While earlier work largely targeted point estimation and/or variable selection, recent years have seen a number of proposals on how to also assign uncertainty, statistical significance and confidence in high-dimensional models. This is of great practical importance, particularly when interpretation of parameters and variables is of key priority. 

Early attempts are highly varied in their approach. Stability
selection was proposed by \cite{meinshausen2010} as a generic method
for controlling the expected number of false positive selections; with
improvements given by \cite{shah2013}. Sample splitting, where a first subsample is used to screen, and a second subsample is used to perform inference  \citep{wasserman2009,meinshausen2009split} has also been explored.
Taking an alternative approach, \cite{lockhart2014}, \cite{tibshirani2014} and \cite{lee2016} build a framework for conditional inference for high-dimensional linear models, i.e., conduct inference given some covariates have been selected. 

In this paper, we propose an \emph{unconditional} inferential framework for high-dimensional linear mixed effect models, with the goal of testing null hypotheses of the form 
\begin{equation}\label{eqn:null}
H_{0, G} ~:~ \beta^*_j = 0 ~\mbox{for all $j \in G$}
\end{equation}
where $\beta^* \in \mathbb R^p$ is the vector of fixed effect
regression coefficients, and $G$ may be any subset of
$\{1, \ldots , p\}$. Of particular interest is the case $G = \{ j \}$,
i.e., testing if a single fixed effect coefficient $\beta^*_j$ is
zero. A related goal is to construct confidence intervals for
$\beta^*_j$, $j = 1, \ldots, p$. This problem arises naturally in many
settings, as observations are rarely independent. A prime example is
the analysis of longitudinal data, which is highly prevalent in
clinical studies. In such settings mixed effect models are a natural
extension of linear models for modeling data exhibiting
group-structured dependence.

Our framework is inspired by a recent line of work that proposes
de-biasing penalized estimators as an approach to inference for
high-dimensional linear models with fixed effects only. There, the
limiting distribution of the modified estimator is tractable and,
thus, can be used to construct approximate $p$-values and confidence
intervals. For example, in high-dimensional linear regression,
\cite{zhang2014}, \cite{vdg2014}, and \cite{javanmard2014} suggest
de-sparsifying the lasso: starting with the biased lasso estimator, the
authors `invert' the corresponding Karuhn-Kush-Tucker (KKT)
optimality conditions to form an estimator that is approximately
unbiased for $\beta^*$ and normally distributed. 
By construction, the de-biased estimator can then be used to derive confidence intervals and $p$-values. \cite{ning2017} extended this strategy by developing a score test for inference in penalized $M$-estimators. 

Our proposed method bears strongest resemblance to
\cite{buhlmann2013}. Developed for high-dimensional linear models, the
framework of \cite{buhlmann2013} is similar to those put
forth by \cite{zhang2014}, \cite{vdg2014}, and \cite{javanmard2014},
except it uses ridge estimation as a starting point. While the overall
framework we consider is similar, there are important differences in the specifics
on how to correct --- or rather, approximately correct --- for the
bias in the ridge estimator, and how to compute an approximation of
the limiting distribution of the de-biased estimator, to
construct $p$-values and confidence intervals for elements in
$\beta^*$. As will be evident later, these differences are 
direct results of having to cope with dependencies induced by the
random effects in the linear mixed effect model. The naive treatment
of ignoring the dependencies, as we demonstrate in numerical examples, leads to poor practical performance (particularly, when inverting estimator to obtain confidence intervals, the confidence intervals have insufficient coverage).  We address this issue by introducing a two-stage procedure that yields consistent estimates of the parameters that determine these dependencies. While we describe a ridge-based framework, the methodology could be extended to make use of other high-dimensional estimators as the starting point for constructing a de-biased estimator.

Our decision to use a ridge estimator is based on simulation findings for standard linear models showing that while asymptotically optimal, confidence intervals from $\ell_1$-based de-biasing \citep{zhang2014, vdg2014, javanmard2014} tend to have coverage problems in finite samples. \cite{yu2018} similarly noticed that confidence intervals based on a de-biased $\ell_1$-estimator for high-dimensional Cox model had poorer than theoretical coverage in practice. Although its theoretical justification is similar, the ridge-based method of \cite{buhlmann2013} yields better finite-sample error control. 

Our paper is organized as follows. The remainder of this section provides a brief overview of the subsequent notation. Section \ref{sec: model} makes explicit the form of the high-dimensional linear mixed effect model we are working with. In Section \ref{sec: method}, we describe the details of our method: specifically, how it builds upon \cite{buhlmann2013} to accommodate dependence within groups induced by the random effects. We also present theory, along with the required assumptions, to justify it.  Numerical experiments can be found in Section \ref{sec: simulations}, followed by a practical application of the method in Section \ref{sec: prac}. We conclude with a discussion and elaborate on potential extensions in Section \ref{sec: discussion}.  
Proofs are collected in the Appendix. 

\subsection*{Notation}
Matrices are written in upper-case bold-face and their entries in corresponding lower-case. So $a_{jk}$ is the $(j,k)$th entry of matrix
$\mathbf A \in \mathbb R^{n_1 \times n_2}$.  For
$j \in \{1, \ldots , n_2\}$ and $J \subseteq \{1, \ldots , n_2\}$, $a_j$ and $\mathbf A_J$ denote the $j$th column of
$\mathbf A$ and the column-wise concatenation of columns in
$\mathbf A$ indexed by the set $J$, respectively.  The $i$th row of
$\mathbf A$ is denoted $a^{(i)}$.
For $r \in [1, \infty]$, the $\ell_r$ norm of a vector
$u \in \mathbb R^{n}$ is
$
\|u\|_r = \big(\sum_{i=1}^p |u_i|^r \big)^{1/r}, 
$
and the induced norm of a matrix $\mathbf A \in \mathbb R^{n_1 \times n_2}$ is 
$
\vertiii{\mathbf A}_r = \sup \left\{\|\mathbf Ax\|_r: x \in \mathbb R^{n_2}, \|x\|_r = 1 \right\}.
$
With this notation, $\vertiii{\mathbf A}_2$ is the spectral norm,  $\vertiii{\mathbf
  A}_1$ the maximum absolute column sum of the matrix, and
$\vertiii{\mathbf A}_\infty$ the maximum absolute row sum of the
matrix. We use $\|\mathbf A\|_r$ to denote the $\ell_r$ norm of the
vectorization of $\mathbf A$.

The projection of $\mathbb R^{n_2}$ onto the linear space generated by
the rows of $\mathbf A$ is denoted
$
\mathbf P_{\mathbf A} = \mathbf A(\mathbf A^T \mathbf A)^{-}\mathbf A^T, 
$
where ${\mathbf A}^-$ is the Moore-Penrose inverse of ${\mathbf A}$.  For square matrices $\mathbf A_1$ and $\mathbf A_2$ of the same dimensions, $\mathbf A_1 \leq \mathbf A_2$ indicates that $\mathbf A_2 - \mathbf A_1$ is positive semi-definite.

For real-valued functions $g_1(x)$ and $g_2(x)$ defined on $(0,\infty)$, we write $g_1(x) \lesssim g_2(x)$
if there is a constant $c \in (0, \infty)$ such that
$g_1(x) \le cg_2(x)$, and $g_1(x) \gtrsim g_2(x)$ if instead $g_1(x) \ge cg_2(x)$.  We write
$g_1(x) \asymp g_2(x)$ if both $g_1(x) \lesssim g_2(x)$ and
$g_1(x) \gtrsim g_2(x)$.  Then, $g_1(x) = o(g_2(x))$ if
$g_1(x)/g_2(x) \to 0$ as $x \to \infty$, and $g_1(x) = O(g_2(x))$ if
there is a $c \in (0, \infty)$ such that
$|g_1(x)|\le c g_2(x)$ for all $x$ large enough. The latter relations
also apply when $x$ is a vector, where $x \to \infty$ is interpreted elementwise.
%
Finally, if $X \in \mathbb R$ is a random variable and $a \in \mathbb R$ is some constant, we write
$
| X - a| = o_P(1) 
$
if $X$ converges to $a$ in probability, i.e., $X \to_p a$.

\section{The linear mixed effect model}\label{sec: model}
Consider $M$ groups of observations of sizes $n_1,\dots,n_M$. Let $m = 1, \ldots , M$ be group indices, and let $i = 1, \ldots , n_m$ index the observations within group $m$. Let $N$ be the total number of observations, so $N = \sum_{m=1}^M n_m$. We may later assume, without loss of generality, that $n_m = n$ for all groups, or that, $N = nM$. The proposed framework allows for non-uniform group sizes with minor adjustments, so long as the group sizes are of the same order. 

For group $m \in \{1, \ldots , M\}$, we observe the response vector $y_m \in \mathbb R^{n}$, generated as
\begin{align}\label{eq: mixed}
y_m &= \mathbf X_m \beta^* + \mathbf Z_m \upsilon_m + \epsilon_m, \quad m = 1, \ldots , M
\end{align}
with
\begin{enumerate}[label=(\roman*)] 
\item $\beta^* \in \mathbb R^p$, an unknown vector of fixed regression coefficients; 
\item $\upsilon_m \in \mathbb R^q$, $m = 1, \ldots , M$ vectors of group-specific random effects, with $\upsilon_m \underset{i.i.d.}{\sim} \mathcal N(0, \mathbf \Psi^*)$, $\mathbf \Psi^*$ an unknown $q \times q$ positive definite covariance matrix; 
\item errors $\epsilon_m \underset{i.i.d.}{\sim} \mathcal N(0, \sigma^{*2}\mathbf I_{n \times n})$ for unknown $\sigma^{*2}$, which are independent of $\upsilon_1, \ldots , \upsilon_M$; and
\item $\mathbf X_m \in \mathbb R ^{n \times p}$ and $\mathbf Z_m \in \mathbb R^{k \times q}$ known design matrices.  
\end{enumerate}
By construction, $\beta^*$ represents effects shared across groups while $\upsilon_m$, $m = 1, \ldots , M$, represent group-specific deviations.  It will be convenient to write the model more compactly.
Define vectors
$
  y = [y_1^T, \ldots, y_M^T]^T,\quad
  \upsilon= [ \upsilon_1^T, \ldots, \upsilon_M^T]^T, \quad
  \epsilon = [\epsilon_1^T, \ldots,\epsilon_M^T]^T,
$
a stacked matrix $\mathbf X = [\mathbf X_1^T, \ldots, \mathbf
X_M^T]^T$, and $\mathbf Z = \mbox{diag}(\mathbf Z_1, \ldots , \mathbf
Z_M)$.  Then we can write (\ref{eq: mixed}) as
\begin{align}\label{eq: vector_mixed}
y = \mathbf X\beta^* + \mathbf Z \upsilon + \epsilon.
\end{align}
Marginalizing out the random effects yields
\begin{align}\label{eq: margin}
y \sim \mathcal N(\mathbf X \beta^*, \mathbf V(\sigma^{*2}, \mathbf \Psi^*)) \quad \mbox{with $\mathbf V(\sigma^{*2}, \mathbf \Psi^*))  = \sigma^{*2}\mathbf I_{N \times N} + \mathbf Z \mathbf \Psi^*) \mathbf Z^T$},
\end{align}
where $\mathbf \Psi^{*(B)} = \mathbf I_{M \times M} \otimes ~\mathbf \Psi^*$. This implies that $\mathbf V(\sigma^{*2}, \mathbf \Psi^*)$ is block-diagonal and observations belonging to different groups are independent. Thus, the inclusion of random effects only induces dependencies between observations belonging to the same group. We will be primarily working with the marginal form (\ref{eq: margin}) in subsequent sections. 

We study the presented model under the following assumptions:
\begin{enumerate}
\item \textit{High dimensions}: We allow $p$, the number of fixed regression coefficients, to be possibly much larger than $N$. On the other hand, $q$, the number of random effect variables, is assumed to be of constant order, or at least smaller than $n$. 
\item {\textit{Sparsity of $\beta^*$}}: We assume $\beta^*$ to be sparse in the sense that most of its elements are zero: a more precise specification on the level of sparsity required is detailed in Section~\ref{sec: consistent}. 
\item \textit{Structure of $\mathbf \Psi^*$}: Our paper primarily considers the scenario of $\mathbf \Psi^* = \tau^{*2}\mathbf I_{q \times q}$. However, our method, and corresponding theoretical results, can be extended to accommodate the more general scenario of $\mathbf \Psi^* = \mathbf D^*$ where $\mathbf D^*$ is a diagonal $q \times q$ matrix. 
\item \textit{Standardization of design matrices}:  The design matrices $\mathbf X$ and $\mathbf Z$ are assumed \textit{fixed} and standardized with $\|x_j\|_2^2 = N$ for $j \in \{1, \ldots , p \}$ and $\|z_j\|_2^2 = n$ for $j \in \{1, \ldots , qM\}$. 
\end{enumerate}

\section{A ridge-based inferential framework}\label{sec: method}
We would like to test null hypotheses of the form \eqref{eqn:null}, i.e., 
$
H_{0, G}: \beta^*_j = 0 ~\mbox{for all $j \in G$}, 
$
for subsets $G\subset\{1,\dots,p\}$, and construct confidence intervals for $\beta^*_j$. 
This section formally introduces our inferential framework. We first
describe its foundation, the de-biased ridge estimator,  and show how it can be used to accomplish these tasks. We then detail how to assemble the components needed to construct this de-biased ridge estimator and approximate its limiting distribution. Theoretical justification of our approach is provided along the way.

\subsection{A de-biased ridge estimator}

As in \cite{buhlmann2013}, our starting point is the ridge estimator given by
\begin{equation}\label{eq: ridge0}
\hat \beta = \mbox{arg}~\underset{\beta \in \mathbb R^p}{\min} ~\| y - \mathbf X\beta \|_2^2/N + \lambda \|\beta\|_2^2.
\end{equation}
This estimator is natural in models with homoscedastic and uncorrelated errors but in the linear mixed effect model, the random effects results in correlation. We thus refer to $\hat\beta$ from (\ref{eq: ridge0}) as the `naive' ridge estimator. The estimator has a simple closed form expression, 
\begin{equation}\label{eq: ridge}
\hat \beta = {N}^{-1}\left(\hat{\mathbf \Sigma} + \lambda \mathbf I_{p \times p}\right)^{-1} \mathbf X^T\mathbf Y, 
\end{equation}
where $\hat{\mathbf \Sigma} = \mathbf X^T\mathbf X/N$. It is straightforward to show that the ridge estimator is normally distributed with covariance matrix, multiplied by a factor of $N$,
\begin{equation}\label{eq: omega}
\mathbf \Omega^* = (\hat{\mathbf \Sigma} + \lambda \mathbf I_{p \times p})^{-1} \mathbf X^T \mathbf V(\sigma^{*2}, \tau^{*2}) \mathbf X (\hat{\mathbf \Sigma} + \lambda \mathbf I_{p \times p})^{-1}/N.
\end{equation}
As in \cite{buhlmann2013}, we assume that the diagonal entries of $\mathbf \Omega^* = ( \omega^*_{jk} )$ satisfy 
\begin{equation}\label{eq: diagonal}
\omega_{\min}^* \equiv \underset{j \in \{1, \ldots , p\}}{\min}~\omega^*_{jj} > 0.
\end{equation}
Likewise, we do not require (\ref{eq: diagonal}) to be bounded away from $0$ as a function of $N$ or $p$. 
This condition, in fact, is fairly mild; it is only violated under special kinds of design matrices. To illustrate, define $R \equiv \mbox{rank}(\mathbf X)$ and let
$\mathbf X = \mathbf Q \mathbf D \mathbf \Gamma^T$
be the singular value decomposition with left singular vectors 
$\mathbf Q \in \mathbb R^{N \times N}$ satisfying $\mathbf Q^T\mathbf Q = \mathbf I_{N \times N}$, $\mathbf D \in \mathbb R^{N \times N}$ a diagonal matrix with entries $s_1 \geq \ldots \geq s_N$ (i.e., singular values of $\mathbf X$), and right singular vectors $\mathbf \Gamma \in \mathbb R^{p \times N}$ satisfying $\mathbf \Gamma^T\mathbf \Gamma = \mathbf I_{N \times N}$. 
Let $\nu_{\min}(\mathbf A)$ and $\nu_{\max}(\mathbf A)$ be the smallest and largest eigenvalue of any square matrix $\mathbf A$, respectively. We can then show the following. 
\begin{lemma}\label{lem:diagonal}
Condition (\ref{eq: diagonal}) holds if and only if $\mathbf X \neq \mathbf 0$ and 
\begin{equation}\label{eq: r2}
\underset{j \in \{1, \ldots , p\}}{\min}~\underset{k \in \{1, \ldots , N\}, s_k \neq 0}{\max}~\mathbf \Gamma_{jk}^2 > 0.
\end{equation}
\end{lemma} 

In the high-dimensional case with $R \leq N < p$, the parameter $\beta^*$ is not identifiable:  many vectors $\theta \in \mathbb R^p$ satisfy $\mathbf X\beta^* = \mathbf X \theta$. 
A natural parameter to consider, as noted in \cite{shao2012}, is $\theta^* = \mathbf P_{\mathbf X^T} \beta^* = \mathbf X^T(\mathbf X\mathbf X^T)^{-}\mathbf X \beta^* = \mathbf \Gamma\mathbf \Gamma^T\beta^*$, the projection of $\beta^*$ onto the linear space generated by the rows of $\mathbf X$. As it turns out, under condition (\ref{eq: diagonal}), or equivalently (\ref{eq: r2}), the ridge estimator $\hat \beta$ is a reasonable proxy for $\theta^*$ when $\lambda$ is sufficiently small. 

\begin{proposition}\label{prop: ridge_theta}
Suppose that $\lambda > 0$ and (\ref{eq: diagonal}), or equivalently, (\ref{eq: r2}), holds. 
Then, under our linear mixed effect model from Section~\ref{sec: model}, the ridge estimator (\ref{eq: ridge}) satisfies 
\begin{align*}
\underset{j \in \{1, \ldots , p\}}{\max} ~ \left| \mathbb E\left[\hat \beta_{j}\right] - \theta^*_{j} \right| &\le \lambda \|\theta^*\|_2 \nu_{\min, +}\left(\hat{\mathbf \Sigma}\right)^{-1} , \\
\underset{j \in \{1, \ldots , p\}}{\min}  ~ \Var\left[\hat \beta_j\right] &\ge N \omega^*_{\min}
\end{align*}
where $\nu_{\min,  +}(\hat{\mathbf \Sigma})$ refers to the smallest non-zero eigenvalue of $\hat{\mathbf \Sigma}$. 
\end{proposition}

Proposition \ref{prop: ridge_theta}, which is proven in the Appendix,
implies that the bias in estimating $\theta^*$ with $\hat \beta$ is
small when $\lambda > 0$ is sufficiently small. We explicitly quantify
how small $\lambda$ needs to be for the estimation bias to be smaller than the standard error of $\hat \beta$. 

\begin{corollary}
Suppose that the ridge penalty parameter $\lambda > 0$ is chosen such that
$
\frac{\lambda}{\sqrt{\omega^*_{\min}}} \le \frac{\nu_{\min, +}(\hat{\mathbf \Sigma})}{\sqrt{N}\|\theta^*\|_2}, 
$
and that condition (\ref{eq: diagonal}), or equivalently, (\ref{eq: r2}) holds. Then, 
\begin{equation*}
\underset{j \in \{1, \ldots, p\}}{\max} ~ \left| \mathbb E\left[\hat \beta_{j}\right] - \theta^*_{j} \right|  \le \underset{j \in \{1, \ldots , p\}}{\min}  ~ \sqrt{\Var\left[\hat \beta_j\right]} .
\end{equation*}
\end{corollary}

Our interest, however, lies in $\beta^*$, not $\theta^*$. Thus, for $\hat \beta$ to be useful, we need to adjust $\hat \beta$ for the projection bias $B_j = \theta_j^* - \beta_j^*$. By definition of $\theta^*$, one observes that
\begin{equation}
B_j = (\mathbf P_{\mathbf X^T} \beta^*)_j - \beta^*_j = (\mathbf P_{\mathbf X^T})_{jj} \beta^*_j - \beta^*_j + \sum_{k \neq j} (\mathbf P_{\mathbf X^T})_{jk}\beta^*_k,
\end{equation}
which, under the null hypothesis $H_{0, j}:\beta_j^* = 0$, becomes,
\begin{equation}\label{eq: proj_bias}
B_{H_0, j} =  \sum_{k \neq j} (\mathbf P_{\mathbf X^T})_{jk}\beta^*_k.
\end{equation}
The quantity can be approximated by
\begin{equation}
\hat B_{H_0, j} =  \sum_{k \neq j} (\mathbf P_{\mathbf X^T})_{jk}\hat \beta^{\mbox{\init}}_k.
\end{equation}
where $\hat \beta^{\mbox{\init}}$ is a consistent initial estimator of
$\beta^*$ (and consistency occurs under additional assumptions).
Consider then the corrected ridge estimator $\hat \beta^{\corr}_{j}$ as a statistic for testing $H_{0, j}$:
\begin{equation}
\hat \beta^{\corr}_{j} = \hat \beta_j - \hat B_{H_0, j} = \hat \beta_j -  \sum_{k \neq j} (\mathbf P_{\mathbf X^T})_{jk}\hat \beta^{\mbox{\init}}_k.
\end{equation}
Assuming that $\underset{j \in \{1, \ldots , p\}}{\min} ~ \mathbf \omega_{\min}^* > 0$, we can write
\[
\hat \beta^{\corr}_{j} = W_j + \gamma_j,
\]
where
\begin{align*}
\gamma_j &= (\mathbf P_{\mathbf X^T})_{jj} \beta_j^* - \sum_{k \neq j} (\mathbf P_{\mathbf X^T})_{jk}\left(\hat \beta^{\init}_k - \beta^*_k\right) + \delta_j, \\
\delta_j &=\delta_j(\lambda)= \mathbb E\left[\hat \beta_j\right] - \theta^*_j.
\end{align*}
A rearrangement of the above set of equations yields
\begin{equation}\label{eq: dist}
\frac{\hat \beta^{\corr}_{j}}{\left(\mathbf P_{\mathbf X^T}\right)_{jj}} - \beta_j^* = \frac{W_j}{(\mathbf P_{\mathbf X^T})_{jj}} - \sum_{k \neq j} \frac{(\mathbf P_{\mathbf X^T})_{jk}}{(\mathbf P_{\mathbf X^T})_{jj}}(\hat \beta^{\init}_k - \beta^*_k) + \frac{\delta_j}{(\mathbf P_{\mathbf X^T})_{jj}}. 
\end{equation}
Then, from model (\ref{eq: margin}), it follows that
\begin{equation}\label{eq: distribution}
W_1, \ldots , W_p \sim \mathcal N(0, \mathbf \Omega^*/N).
\end{equation}

The normalizing factors needed to bring the $W_j$ to $N(0, 1)$ scale are given by $\kappa_{j} = \kappa_j(N, p) = \sqrt{N/\omega_{jj}^*}$. The proof is straightforward.

\begin{theorem}\label{thm: pvalue}
Suppose we choose the ridge penalty parameter $\lambda > 0$ such that
\begin{equation}\label{eq: lambda_ridge_choice}
\lambda \left(\omega^*_{\min}\right)^{-1/2} = o\left(\nu_{\min, +}\big(\hat{\mathbf \Sigma}\big)/\left(N^{1/2}\|\theta^*\|_2\right)\right),
 \quad (N, p \to \infty),
\end{equation}
and assume that for our choice of $\hat \beta^{\init}$, there exist constants $C_j = C_j(N, p)$ such that 
\begin{equation}\label{eq: c_j}
\mathbb P\left[\bigcap_{j=1}^p \left\{\left| \kappa_j(N, p)\sum_{k \neq j} (\mathbf P_{\mathbf X^T})_{jk}\left(\hat \beta^{\init}_k - \beta^*_k\right) \right| \le C_j(N, p) \right\} \right] \to 1 \quad (N, p \to \infty).
\end{equation}
Then, under the null hypothesis, $H_{0, j}$, for all $w > 0$,
\begin{equation}
\underset{N, p \to \infty}{\lim \sup} ~\mathbb P\left[ \left|\kappa_j\hat\beta^{\corr}_j \right| > w \right] - \mathbb P\left[|\widetilde W| + C_j > w\right] \le 0,
\end{equation}
where $\widetilde W \sim N(0, 1)$. In addition, for any sequence of subsets $G_p \subseteq \{1, \ldots , p\}$, if $H_{0, G_p}$ is true, then for any $w > 0$,
\begin{equation}
\underset{N, p \to \infty}{\lim \sup} ~\mathbb P\left[\max_{j \in G_p}~ \left|\kappa_j\hat \beta^{\corr}_j \right| > w \right] - \mathbb P\left[\max_{j \in G_p}~ \left(|\widetilde W| + C_j\right) > w\right] \le 0. 
\end{equation}
\end{theorem}

In subsequent sections, we identify specific scalings of $N$ and $p$
such that Theorem \ref{thm: pvalue} becomes applicable.  Based on the asymptotic distributions in Theorem \ref{thm: pvalue}, we can construct $p$-values for testing $H_{0,G}$, $G \subseteq \{1, \ldots , p\}$. For testing the individual null hypothesis $H_{0,j}$, we define the $p$-value for the two-sided alternative as
\begin{equation}\label{eq: pj}
\varrho_j = 2(1 - \Phi((\kappa_j|\hat \beta^{\corr}_j| - C_j)_+)),
\end{equation}
where $\Phi$ is the standard normal distribution function. For
testing the group null hypothesis $H_{0, G}$, $|G| > 1$, we define the
$p$-value as 
\begin{equation}\label{eq: pg}
\varrho_G = 1 - \mathbb P\left[\underset{j \in G}{\max}~\left(\kappa_j|W_j| + C_j\right) \le \underset{j \in G}{\max}~\kappa_j|\hat \beta^{\corr}_j| \right],
\end{equation}
where $W_1, \ldots , W_p$ are as in  (\ref{eq: distribution}). From Theorem \ref{thm: pvalue}, we can derive the following corollary. 

\begin{corollary}\label{cor: pvalue}
Under the conditions in Theorem \ref{thm: pvalue}, for any $\alpha \in (0, 1)$, the following statements hold:
\begin{align*}
\underset{N, p \to \infty}{\lim \sup} ~\mathbb P\left[\varrho_j \le \alpha\right] - \alpha &\le 0 \quad \mbox{if $H_{0,j}$ is true}, \\ 
\underset{N, p \to \infty}{\lim \sup} ~\mathbb P\left[\varrho_G \le \alpha\right] - \alpha &\le 0 \quad \mbox{if $H_{0,G}$ is true}. 
\end{align*}
\end{corollary}

\subsection{Consistent estimation of variance parameters}\label{sec: consistent}

As presented, the de-biased ridge framework depends on the values of the unknown parameters $\sigma^{*2}$ an $\tau^{*2}$. We employ a two-step approach to consistent estimation of these parameters.

\begin{enumerate}
\item Let $S = \{j ~:~ \beta^*_j \neq 0 \}$ be the support of
  $\beta^*$,  with cardinality $d = |S|$.  We use the Lasso estimator
  $\hat \beta^L = \mbox{arg}~\underset{\beta \in \mathbb R^p}{\min}
  ~\| y - \mathbf X\beta \|_2^2/N + 2\lambda_L \|\beta\|_1$
  with an appropriate choice of tuning parameter $\lambda_L$ to identify an initial guess of the elements (i.e., indices) in $S$. We define $\hat S = \{j ~:~ \hat \beta^L_j \neq 0\}$ as our guess for the support $S$. By properties of the Lasso, $|\hat S| \le N$, although, in general, $\hat S$ may not be a good estimate of $S$.
\item Working with the (potentially misspecified) random effects model
\begin{equation}\label{eq: wrong}
y = \mathbf X_{\hat S}\beta_{\hat S}^* + \mathbf Zb + \epsilon,
\end{equation}
we apply Henderson's Method III \citep{henderson1953} to form
estimates $\hat \sigma^2$ and $\hat \tau^2$.
Henderson's Method III is particularly
tractable theoretically and enables us to study
consistency in the scenario where (\ref{eq: wrong}) is actually
misspecified, i.e., $|S \backslash \hat S| > 0$.  For a discussion of
Henderson's methods and the appeals of Method III, see \citep{searle1968}.
\end{enumerate}

In recent years Henderson's methods have largely been supplanted by
alternatives such as restricted maximum likelihood (REML) for variance
component estimation \citep{harville1977};  it is customary to refer
to variances of random effects as variance components.
We thus provide a brief overview of what Henderson's Method III entails.
Consider, first, the low-dimensional model (\ref{eq: margin}) with $p < N$. To simplify the notation in the following explanation, we momentarily define $\mathbf{\tilde X} = \begin{bmatrix}\mathbf X & \mathbf Z\end{bmatrix}$. By not distinguishing between fixed and random effects, the idea behind Henderson's methods is to match the differences in the reductions in the sum-of-squares between sub-models of (\ref{eq: margin}) to its expected value, not unlike a method-of-moments approach. To elaborate, in fitting (\ref{eq: margin}) to data $y$, the reduction in the sum of squares is
\begin{equation}
\mathcal R(\beta, \upsilon) = y^T\mathbf P_{\mathbf{\tilde X}} y.
\end{equation}
Likewise, the decrease in the sum of squares due to fitting the reduced model $y = \mathbf X\beta + \epsilon$ is
\begin{equation}
\mathcal R(\beta) = y^T\mathbf P_{\mathbf X} y.
\end{equation}
The expected difference in the reductions $\mathcal R(\upsilon|\beta) \equiv \mathcal R(\beta, \upsilon) - \mathcal R(\beta)$ is
\begin{align}\label{eq: hend1}
\mathbb E[\mathcal R(\upsilon|\beta)] = \tau^{*2} \mbox{tr}\left(\mathbf Z^T\left[\mathbf I_{N \times N} -  \mathbf P_{\mathbf X}\right] \mathbf Z\right) + \sigma^{*2}\left[\mbox{rank}\left(\mathbf{\tilde X}\right) - \mbox{rank}\left(\mathbf X\right)\right].
\end{align}
Moreover,
\begin{equation}\label{eq: hend2}
\mathbb E\left[y^Ty - R(\beta, \upsilon)\right] =  \sigma^{*2}\left[N - \mbox{rank}\left(\mathbf{\tilde X}\right)\right].
\end{equation}
Together, (\ref{eq: hend1}) and (\ref{eq: hend2}), when matching
theoretical expectations to empirical averages, form a triangular
system of linear equations, from which we derive $\hat \sigma^2$ and
$\hat \tau^2$.  
We find
\begin{align}
\hat\sigma^2 &= \frac{y^T\left(\mathbf I_{N\times N} - \mathbf P_{\mathbf{\tilde X}}\right)y}{N - \mbox{rank}\left(\mathbf {\tilde X}\right)}, \label{eq: sigma_hat} \\
\hat \tau^2 &= \frac{y^T\left(\mathbf P_{\mathbf{\tilde X}}- \mathbf P_{\mathbf X}\right)y - \hat\sigma^2\left[\mbox{rank}\left(\mathbf{\tilde X}\right) - \mbox{rank}(\mathbf X)\right]}{ \mbox{tr}(\mathbf Z^T(\mathbf I_{N \times N} -  \mathbf P_{\mathbf X}) \mathbf Z)}. \label{eq: tau_hat}
\end{align}
It is straightforward to see that the $\hat \sigma^2$ and $\hat \tau^2$ generated from (\ref{eq: sigma_hat}) and (\ref{eq: tau_hat}) are unbiased, presuming that the true model is $y = \mathbf X\beta + \mathbf Z \upsilon + \epsilon$. For consistency, some additional assumptions are needed, which we will discuss later in this section. 

Returning to our two-step procedure and high-dimensional setup, Step 1
identifies a candidate low-dimensional sub-model, which is used in
Step 2 to obtain variance component estimates. We do not require the
candidate model to encompass the truth; however, $\lambda_L$ should be
such that $\hat S$, from Step 1, reliably captures the indices of the
`strong' signals in $\beta^*$.  The idea is that missing `weak'
signals only negligibly affect the accuracy of $\hat \sigma^2$ and
$\hat \tau^2$ in Step 2. We now show that this two-step procedure yields consistent estimators
$\hat \sigma^2$ and $\hat \tau^2$
in the setting where $N \to \infty$ (specifically, $n$ is fixed, but
the number of groups $M \to \infty$) and $d^2\log p/M = o(1)$,
provided some additional technical assumptions hold. From here on,
this will also be the scaling assumed for Theorem \ref{thm: pvalue},
as well as Corollary \ref{cor: pvalue}. We first present the
assumptions necessary for consistency and then formally state the theorem. 

For $\xi > 1$, define the cone
\begin{equation}\label{eq: cone}
\mathcal C(\xi, S) = \{u \in \mathbb R^p~:~\|u_{S^c}\|_1 \le \xi\|u_S\|_1\}. 
\end{equation}
\begin{assumption} \label{a: filter} 
For some constant $\xi > 1$, 
\begin{equation}\label{eq: zeta}
\zeta \equiv \inf \left\{\frac{\|\mathbf{\hat \Sigma} u\|_\infty}{\|u_A\|_\infty} : u \in \mathcal C_-(\xi, S), |A \backslash S| \le p \right\} \gtrsim 1
\end{equation}
with
$
\mathcal C_-(\xi, S) \equiv \{u  : u \in \mathcal C(\xi, S), ~ u_j\mathbf \Sigma_{j, \cdot} u \le 0 ~\forall j \notin S\}
$
the sign-restricted version of (\ref{eq: cone}).
\end{assumption}
The quantity $\zeta$ in (\ref{eq: zeta}) 
is defined more generally in \cite{ye2010}, where it is termed a sign-restricted cone invertibility factor (SCIF). We have the following lemma.

\begin{lemma}\label{lem: shat}
Suppose Assumption \ref{a: filter} holds, and let $\lambda_L$ be defined by (\ref{eq: lasso_lambda}) (or \ref{eq: lambdaL})
for some small $\varepsilon > 0$ and $\xi$ as in Assumption \ref{a: filter}. If $u^* \le \lambda_L(\xi - 1)/(\xi + 1)$, then
\begin{equation}
\| \hat\beta^L - \beta^*\|_\infty \le \frac{\lambda_L + u^*}{\zeta} \le \frac{2\xi\lambda_L}{(\xi+1)\zeta}. 
\end{equation}
\end{lemma}

In the proof of Lemma \ref{lem: shat} (provided in the Appendix),  SCIF naturally appears when deriving an upper bound for $\|\hat \beta^L - \beta^*\|_\infty$. Lemma \ref{lem: shat} assumes that Assumption \ref{a: filter} is satisfied, and that $\lambda_L$ in Step~1 is chosen such that 
\begin{equation}\label{eq: lambdaL}
\lambda_L = \frac{(\xi + 1)}{(\xi-1)}\sqrt{\frac{2(\sigma^{*2} + \tau^{*2} qn)(\log p - \log(\varepsilon/2))}{N}}  \asymp \sqrt{\frac{ \log p}{qM}} = o(1),
\end{equation}
with $\xi$ as in Assumption \ref{a: filter}. It then establishes that 
\[
	\|\hat \beta^L - \beta^*\|_\infty \le 2\xi\lambda_L/\zeta(\xi + 1) = o(1), 
\] 
with probability exceeding $1 - \varepsilon$, where $\varepsilon > 0$ can be taken arbitrarily small.  A direct implication is that if the lemma's conditions are satisfied, $S \backslash \hat S$ only includes indices corresponding to `weak' signals in $\beta^*$ of magnitude less than $4\xi\lambda_L/\zeta(\xi + 1) = o(1)$ with close to certainty, which is part of what Step 1 sets out to achieve. 

\begin{assumption} \label{a: nprime} 
There exists an integer $N' \lesssim d$ such that for the same constant $\xi > 1$ as in Assumption \ref{a: filter},
\begin{equation}
\frac{d\xi^2}{\psi^2(\xi, S)} < \frac{N'}{\psi_+(N', S)},
\end{equation}
where 
\begin{equation}\label{eq: compat}
\psi(\xi, S) = \min \left\{\frac{d^{1/2}\|\mathbf Xu\|_2}{N^{1/2}\|u_S\|_1}~:~u \in \mathcal C(\xi, S), u \neq 0 \right\}
\end{equation}
and 
$
\psi_+(N', S) = \underset{\mathcal A \cap S = \emptyset, |A | \le N'}{\max}~\nu_{\min}\left(\frac{\mathbf X^T_{\mathcal A}\mathbf X_{\mathcal A}}{N} \right)
$
is the sparse upper eigenvalue of models disjoint with $S$. 
\end{assumption} 
Assumption \ref{a: nprime} is needed to control the number of false positive selections in $\hat S$ from Step 1. In particular, we have
\begin{lemma}\label{lem: fp}
Suppose that Assumption \ref{a: nprime} holds, and $\lambda_L$ is defined according to (\ref{eq: lambdaL}). 
In the event that $u^* \le \lambda_L(\xi - 1)/(\xi + 1)$,  $|\hat S \backslash S| < N'$.
\end{lemma}
Put simply, Lemma \ref{lem: fp} claims that under Assumption \ref{a: nprime} and our choice of $\lambda_L$ from (\ref{eq: lambdaL}), the total number of false selections in Step 1 is bounded by $N'$, with probability exceeding $1 - \varepsilon$. The proof is provided in the Appendix. 

\begin{assumption}\label{a: z}
Let $\check{\mathbf X}$ be formed by joining any $N'$ columns in $\mathbf X$ with $\beta_j^* = 0$ to the $d$ support columns in $\mathbf X$. For the same $N'$ as in Assumption~\ref{a: nprime}, 
\begin{align}
&\mbox{rank}([\mathbf I_{N \times N} - \mathbf P_{\check{\mathbf X}}]\mathbf Z) = \mbox{rank}(\mathbf Z) = qM, \label{eq: z_rank}\\
&\mathbf Z^T[\mathbf I_{N \times N} - \mathbf P_{\check{\mathbf X}}]\mathbf Z \gtrsim \mathbf I_{qM \times qM}, \label{eq: z_dist}
\end{align}
and the $qM$ singular values of $[\mathbf I_{N \times N} - \mathbf P_{\check{\mathbf X}}]\mathbf Z$, $s_1, \ldots , s_{qM}$, satisfy
\begin{align}
\frac{\left|\{i: ~ s_i \neq 0\}\right|}{\left(\sum_{i = 1}^{qM} s_i^2\right)^2} = o(1). \label{eq: z_sing}
\end{align}
\end{assumption}
By (\ref{eq: z_rank}) in Assumption~\ref{a: z}, the fixed data matrix $\mathbf Z$ has full column rank, and no column vector of $\mathbf Z$ can be represented as a linear combination of the column vectors of any `feasible' $\mathbf X_{\hat S}$, assuming that $\lambda_L$ is chosen according to (\ref{eq: lambdaL}). After all, $N' + d$ is the upper bound on the number of selected fixed effects with probability exceeding $1 - \varepsilon$ (Lemma \ref{lem: fp}). Additionally, by (\ref{eq: z_dist}), the sum of the squared perpendicular distances between each column vector in $\mathbf Z$ and its projection onto the linear subspace spanned by the column vectors of feasible $\mathbf X_{\hat S}$' matrices is at least on the order of $qM$ (substantial, given there are $qM$ columns in $\mathbf Z$). The latter half of Assumption~\ref{a: z} requires all columns of $(\mathbf I_{N \times N} - \mathbf P_{\check{\mathbf X}})\mathbf Z$ are `close' to being linearly independent from one another and `contribute equally' to its rank. In particular, note that (\ref{eq: z_sing}) is satisfied if
\begin{equation}
 c_{1} < \frac{s_j}{s_k} < c_{2} \quad \mbox{for $j \neq k$ and some constants $c_{1}, c_{2} > 0$}.
\end{equation}

It is thus clear that (\ref{eq: z_rank}) and (\ref{eq: z_dist}) imply that random effects must not be confounded with any `feasible' set of fixed effects (from Step~1) while (\ref{eq: z_sing}) implies that the random effects are not confounded from one another. Analogous conditions were shown to be necessary to prove consistency of REML estimators in \cite{jiang1996}. 

\begin{assumption}\label{a: something}
For any $j \in S$ such that $|\beta^*_j| < 4\xi\lambda_L/\zeta(\xi + 1)$, with $\lambda_L$ defined as in (\ref{eq: lambdaL}),
$
\|\mathbf \Gamma_{\tilde{\mathbf X}}x_j\|_\infty \asymp 1.
$
Here, $\mathbf \Gamma_{\tilde{\mathbf X}}\mathbf D_{\tilde{\mathbf X}}\mathbf \Gamma_{\tilde{\mathbf X}}^T$ is the eigen-decomposition of $\tilde{\mathbf X}\left(\tilde{\mathbf X}^T\tilde{\mathbf X}\right)^{-}\tilde{\mathbf X}^T$ (defined for this Assumption) with $\tilde{\mathbf X} = \begin{bmatrix} \mathbf {\check X} & \mathbf Z \end{bmatrix}$, where $\mathbf {\check X}$ is formed by joining any $N'$ columns in $\mathbf X$ with $\beta_j^* = 0$ to the $d-1$ support (excluding $j$) columns in $\mathbf X$. The $N'$ referenced here is the same as in Assumptions \ref{a: nprime} and \ref{a: z}.
\end{assumption}
Assumption~\ref{a: something} requires that covariates corresponding to weak (but non-zero) signals in $\beta^*$ (for which we cannot quantify a bound on the probability they are to be included in $\mathbf X_{\hat S}$) are not too strongly correlated to covariates in $\mathbf X_{\hat S}$ nor covariates associated with the random effects. This somewhat resembles the irrepresentability conditions needed for model selection consistency in Lasso---see, e.g., \cite{zhao2006}. However, the two assumptions are very different: Aside from differences in the quantities involved, a key difference is that the irrepresentability condition requires a very stringent upper bound on non-confounding between fixed effects, whereas Assumption~\ref{a: something} only requires boundedness. As shown in the numerical experiments in the Appendix, as the number of covariates and sparsity of the model vary, Assumption~\ref{a: something} is very likely to be satisfied with even small bounds, whereas the irrepresentability condition is increasingly less likely to hold. 

We can now state our main result on consistency of variance component
estimators, which validates our two-step procedure. 
\begin{theorem}\label{thm: main_var}
Consider $N, p \to \infty$ with $n$ fixed, $M \to \infty$. Furthermore, suppose $p \to \infty$ with $d^2q \log p/M = o(1)$. Suppose Assumptions \ref{a: filter}-\ref{a: something} are satisfied and $\lambda_L$ is chosen according to (\ref{eq: lambdaL}) with $\varepsilon ~\propto~ 1/p$. Then, $\hat \sigma^2$ and $\hat \tau^2$ are consistent for $\sigma^{*2}$ and $\tau^{*2}$, respectively, i.e.,
\begin{align}
|\hat \sigma^2 - \sigma^{*2}| = |\hat \tau^2 - \tau^{*2}| = o_P(1) \quad (N, p \to \infty).
\end{align}
\end{theorem}

Because $|\hat \sigma^2 - \sigma^{*2}|$ and $|\hat \tau^2 - \tau^{*2}|$ are both $o_P(1)$, we can use $\hat \sigma^2$ and $\hat \tau^2$ as plug-in values for $\sigma^{*2}$ and $\tau^{*2}$, respectively. From there, we can form a consistent estimator of $\mathbf \Omega^*$ and normalizing constants $\kappa_j$.

For practical applications, REML can be used as a substitute for
Henderson's Method III for Step 2.
Theory for REML would be a possible avenue for further explorations.

\subsection{An initial estimator for $\beta^*$ and our choice of $C_j$}
To form $\hat \beta^{\init}$, we consider the ordinary least-squares (OLS) fit restricted to $\hat S$, i.e.,
\begin{equation}
\hat \beta^{\init} = \mbox{arg}~\underset{\beta \in \mathbb R^p : \beta_{\hat S^c} = 0}{\min} \| y - \mathbf X \beta \|_2^2.
\end{equation}
We proceed to demonstrate that the error $\hat \beta^{\init} - \beta^*$ is $o(1)$ in $\ell_1$ norm. 
\begin{assumption}\label{a: ols}
For the same $N'$ as in Assumptions \ref{a: nprime}, \ref{a: z}, \ref{a: something}, the sparse lower eigenvalue for models containing $S$ of cardinality smaller than $d + N'$ is constant and greater than $0$,
\[
\psi_-(N', S) = \underset{\mathcal A \supset S , | \mathcal A  \backslash S| \le N'}{\min}~\nu_{\min}\left(\frac{\mathbf X^T_{\mathcal A}\mathbf X_{\mathcal A}}{N} \right) \gtrsim 1,
\]
\end{assumption}
Assumption \ref{a: ols}, in conjunction with previous assumptions and choice of $\lambda_L$ (\ref{eq: lambdaL}), can be used to control the $\ell_1$ norm of the estimation error $\hat \beta^{\init} - \beta^*$. 

\begin{theorem}\label{thm: ols} 
Suppose Assumptions \ref{a: filter}--\ref{a: ols} hold. Under
the same conditions as in Theorem \ref{thm: main_var}, for some universal constant $C > 0$,
\begin{equation}
\| \hat \beta^{\init} - \beta^*\|_1 \le Cd\sqrt{\frac{q\log p}{M}}
\end{equation}
 with probability converging to $1$ as $N, p \to \infty$.  
\end{theorem}

Theorem \ref{thm: ols} implies that we have the following crude bound, based on H\"{o}lder's inequality,
\begin{align}
\left| \kappa_j \sum_{k \neq j} \left(\mathbf P_{\mathbf X^T}\right)_{jk}\left(\hat \beta^{\init}_k - \beta_k^*\right)\right| &\le 
\kappa_j \max_{k \neq j}\left|\left(\mathbf P_{\mathbf X^T}\right)_{jk}\right|\| \hat \beta^{\init} - \beta^*\|_1 \nonumber \\
& \le \kappa_j \max_{k \neq j}\left|\left(\mathbf P_{\mathbf X^T}\right)_{jk}\right| Cd \lambda_L. \label{eq: crude}
\end{align}
The following corollary is a direct consequence of the crude bound (\ref{eq: crude}).
\begin{corollary}\label{cor: cj_choice}
Suppose the conditions in Theorem \ref{thm: ols} are satisfied, and that $d$, the sparsity of $\beta^*$, satisfies
$
d \le C^{-1}\left(M / (q \log p)\right)^{\eta}, 
$
with $C$ as in Theorem \ref{thm: ols} and $\eta \in (0, 1/2)$. Then, 
\begin{equation}
C_j = \underset{k \neq j}{\max}~|\kappa_j (\mathbf P_{\mathbf X^T})_{jk}| \left(\frac{q \log p}{M} \right)^{1/2 - \eta}
\end{equation}
satisfies condition (\ref{eq: c_j}) in Theorem~\ref{thm: pvalue}. 
\end{corollary}

\section{Numerical experiments}\label{sec: simulations}

\subsection{A practical choice for $\lambda_L$} 

In practical applications, we run into the issue of not being able to set $\lambda_L$ according to (\ref{eq: lambdaL}), as it involves knowing $\tau^*$ and $\sigma^*$. However, we can derive a (slightly ad-hoc) approximation of what $\lambda_L$ should be. Upon closer examination of the proof of Lemma \ref{lem: conc}, we can substitute the term $\sigma^{*2} + \tau^{*2}qn$ with $\nu_{\max}(\mathbf V(\sigma^*, \tau^*)) = \sigma^{*2} + \tau^{*2}\nu_{\max}(\mathbf Z^T\mathbf Z)$. The latter can be approximated according to the following procedure, assuming that the ratio $\tau^{*}/\sigma^{*}$ is not too small:

\begin{enumerate}
\item Apply scaled lasso \citep{sun2012} to obtain an initial `average' noise estimate. The solution to the scaled lasso problem is characterized by
\begin{equation}
(\hat \beta^{\mbox{{\tiny scaled}}}, \hat \sigma^{\mbox{{\tiny scaled}}}) \in \mbox{arg}~\underset{\beta, \sigma}{\min}~ \frac{\|y - \mathbf X \beta\|_2^2}{2N\sigma} + \frac{\sigma}{2} + \lambda_{\mbox{{\tiny univ}}} \|\beta\|_1
\end{equation}
with $\lambda_{\mbox{{\tiny univ}}} = \sqrt{2 \log p/N}$. 
\item Take $\lambda_L =\hat \sigma^{\mbox{{\tiny scaled}}}\lambda_{\mbox{{\tiny
        univ}}}
 \rho_Z$ with
\begin{equation}\label{eq: adjust}
 \rho_Z  = \sqrt{\frac{\nu_{\max}(\mathbf Z^T\mathbf Z)}{\mbox{tr}(\mathbf Z^T \mathbf Z)/N}}.
\end{equation}
\end{enumerate} 
We provide a heuristic justification. Ignoring the finer details involved in the theory, for the scaled lasso, $(\hat \sigma^{\mbox{{\tiny scaled}}})^2$ serves as a good approximation for $\| \epsilon^* \|^2_2/N$,
where we have defined $\epsilon^* = y - \mathbf X \beta^*$. In linear models, $\epsilon^*$ holds i.i.d.~observations drawn from a $N(0, \sigma^{*2})$ distribution. By the law of large numbers,  $\| \epsilon^* \|^2_2/N$ converges to $\sigma^{*2}$ for large $N$. 
Under a heteroskedastic error model, with  $\epsilon^*$ independent and $\epsilon_i^* \sim N(0, \sigma^{*2}_i)$, we can match $\| \epsilon^* \|^2_2/N $ to its expectation, which is given by $\sum_{i=1}^N \sigma^{*2}_i/N$, so $(\hat \sigma^{\mbox{{\tiny scaled}}})^2$ can be used to approximate the `average' noise level. If $\epsilon^* \sim N\left(0, \mathbf V(\sigma^*, \tau^*)\right)$, then using a similar expectation matching argument, we can expect $(\hat \sigma^{\mbox{{\tiny scaled}}})^2$ to act as a surrogate for
\begin{equation}
\sigma^{*2} + \frac{\tau^{*2}\mbox{tr}(\mathbf Z\mathbf Z^T)}{N},
\end{equation}
which follows from the fact that $\|\mathbf \Gamma \epsilon^*\|_2 =
\|\epsilon^*\|_2$ for any $N \times N$ orthogonal matrix $\mathbf
\Gamma$ (overloading $\mathbf \Gamma$ from (\ref{eq: r2})). What we
actually need is $\sigma^{*2} + \tau^{*2}\nu_{\max}(\mathbf Z^T\mathbf
Z)$. Then in the scenario where ratio $\tau^{*2}/\sigma^{*2}$ is
not too small, $\rho_Z$ from (\ref{eq: adjust}) should give us a choice of $\lambda_L$ that is close to the desired one from (\ref{eq: lambdaL}). Our choice of $\lambda_L$ is constructed according to the above procedure for all subsequent numerical experiments.

\subsection{A look into $p$-values} \label{sec: pval}
Denote the `unblocked' version of $\mathbf Z$ as $\mathbf Z_u$; i.e., $\mathbf Z_u$ is a $N \times q$ matrix formed by row-wise concatenating the $M$ diagonal blocks in $\mathbf Z$. We generate data from model (\ref{eq: mixed}) according to following schemes, setting $M = 25$ and $n = 6$:
\begin{enumerate}
\item[(M1)] For $p \in \{300, 600\}$, $q \in \{1, 2\}$, we construct
  $\begin{bmatrix} \mathbf X & \mathbf Z_u \end{bmatrix}$ from $N$
  i.i.d.~realizations from a $\mathcal N(0, \mathbf \Phi^*)$
  distribution with $\mathbf \Phi^* = \{ \phi_{jk} \}$ a
  $(p + q) \times (p + q)$ matrix with $\phi^*_{jk} =
  0.2^{|j-k|}$. $\mathbf X$ and $\mathbf Z$ (the `blocked' version)
  are then normalized such that $\|x_j\|_2^2 = N$ and
  $\|z_j\|_2^2 = n$ for all $j$. For $b \in \{0.5, 1\}$, we set the $p$-dimensional vector of fixed regression coefficients to
\begin{equation*}
\beta = [b, \ldots , b, 0, \ldots , 0],
\end{equation*}
where, the first $d \in \{5, 10\}$ entries of $\beta$ are nonzero. The variance parameters $\sigma^{*}$
and $\tau^{*}$ are set to $0.5$ and $1$ respectively.
\item[(M2)] \label{item: m2} Same as (M1) except with $\mathbf \Phi^* = \mathbf I_{(p+q) \times (p+q)}$.
\end{enumerate}
The numerical experiments are setup similarly to those in
\cite{buhlmann2013} and \cite{schelldorfer2011}.  We set the ridge penalty parameter $\lambda$ to $1/N$ for all experiments. Additionally, we set $C_j$ according to Corollary \ref{cor: cj_choice} with $\eta = 0.005$. 

We first consider null hypotheses of the form
\begin{align}
H_{0, j}: \beta_j &= 0. \label{eq: single_hyp} 
\end{align}
We consider decision rules based on a significance level $\alpha = 0.05$, i.e.,
we reject $H_{0,j}$ if the event $E_j=\{\varrho_j \le
0.05\}$ occurs, where $\varrho_j$ is as defined in (\ref{eq: pj}). 
Following \cite{buhlmann2013}, we evaluate the performance of the tests based on the type I error, averaged over the non-support indices,
\begin{equation}
\mbox{Avg. type I error} = (p - d)^{-1} \sum_{j \in S^c} \hat{\mathbb P}(E_j),
\end{equation}
and the power, averaged over the support indices,
\begin{equation}
\mbox{Avg. power} = d^{-1} \sum_{j \in S} \hat{\mathbb P}(E_j),
\end{equation}
where $\hat{\mathbb P}$ denotes the empirical probability over 1000 simulations.  The results, presented in Figure~\ref{fig: individual_tests}, suggest that type I error is well-controlled for all combinations of $p$, $q$, $b$ and $d$ for the two different models. Power is high in most scenarios, but appears to vary with the aforementioned quantities, noticeably decreasing with $b$. However, this is to be expected.

\begin{figure}[t]
    \centering
    \begin{subfigure}[b]{0.47\textwidth}
        \includegraphics[width=\textwidth]{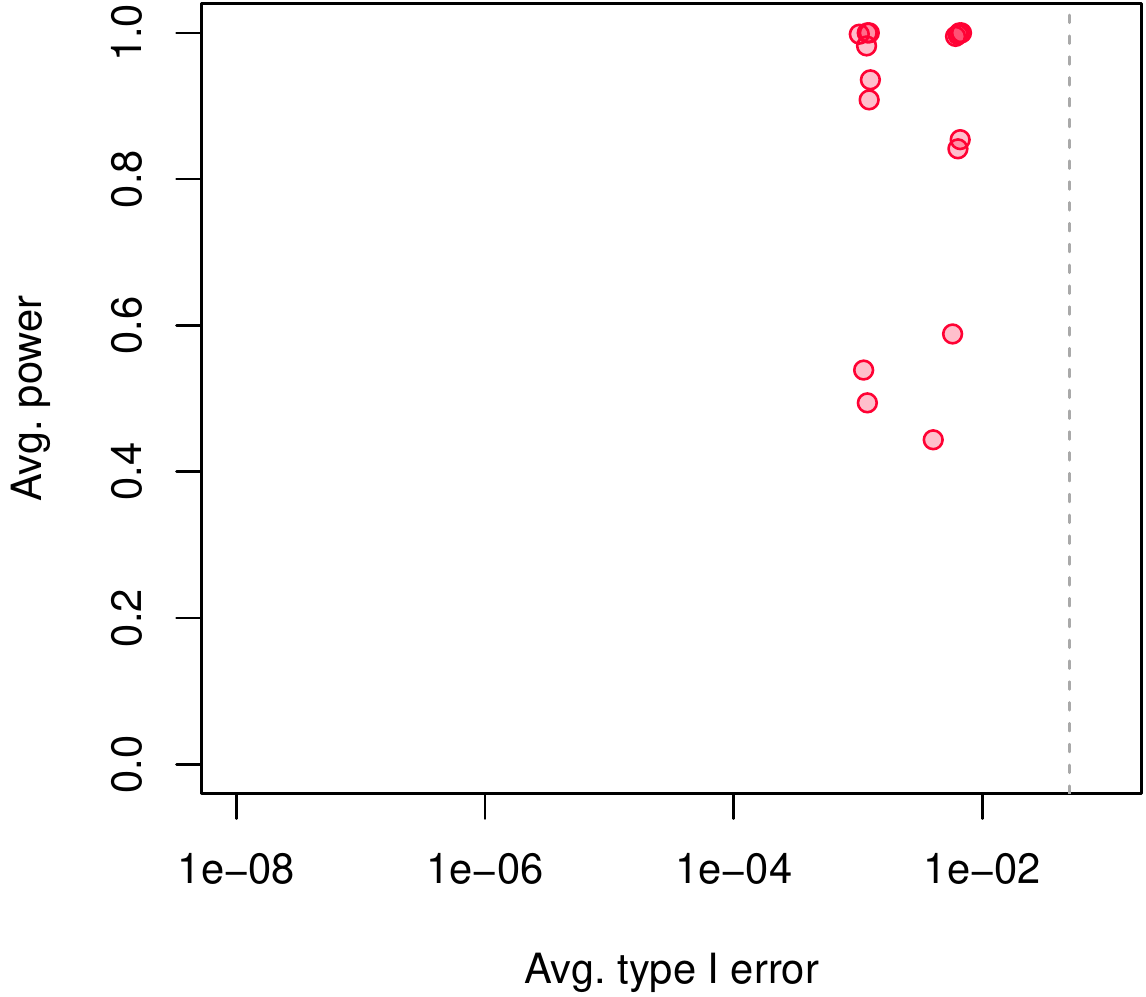}
        \caption{M1}
    \end{subfigure}
    ~ 
    \begin{subfigure}[b]{0.47\textwidth}
        \includegraphics[width=\textwidth]{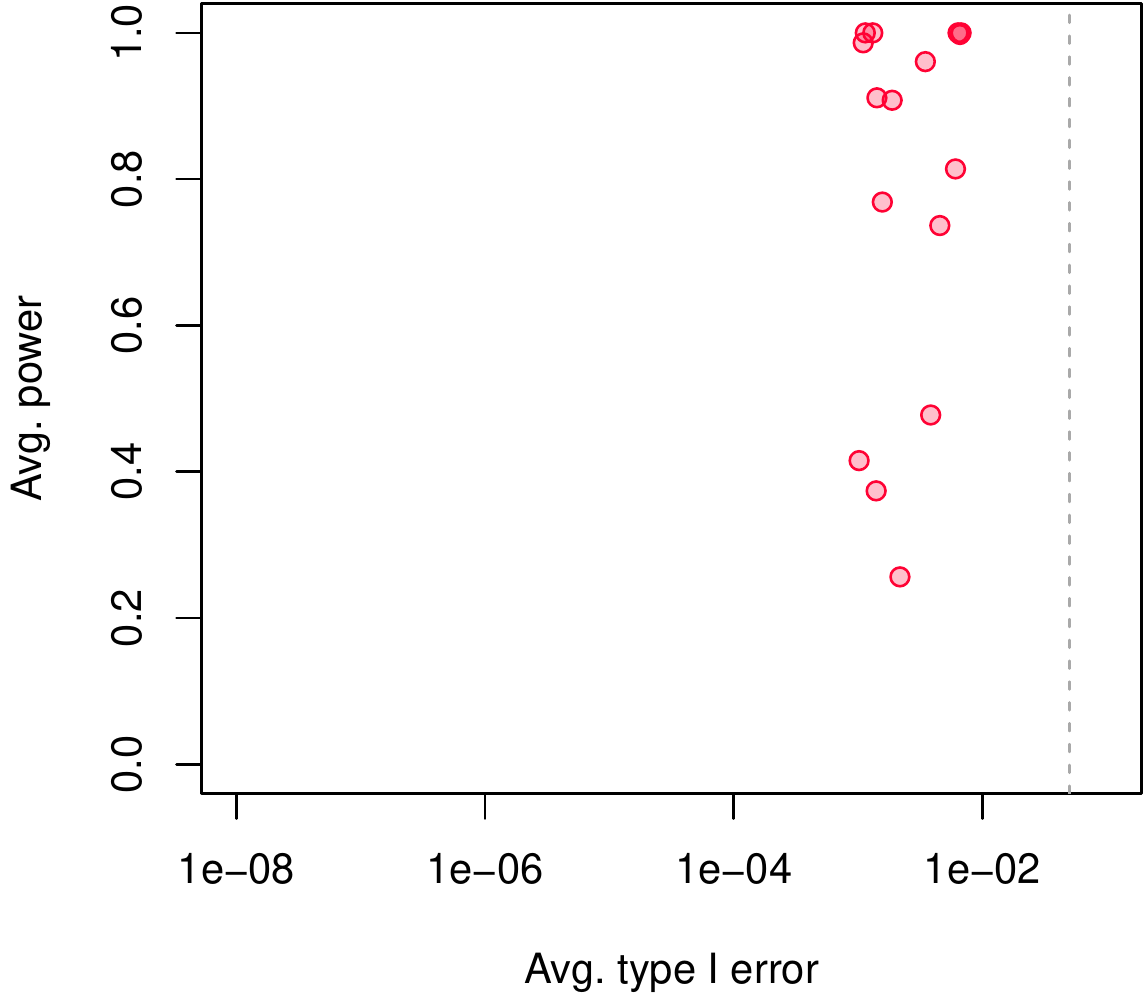}
        \caption{M2}
    \end{subfigure}
    \caption{Average power vs. average type I error for testing groups of coefficients under the two models for different combinations of $p$, $q$, $b$ and $d$.}\label{fig: individual_tests}
\end{figure}

We also consider null hypotheses of the form
\begin{align}
H_{0, G}: \beta_j &= 0 ~\mbox{for all}~ j \in G.  \label{eq: multi_hyp}
\end{align}
with $G$ taken either to be $\{1, \ldots , 100\}$ ($G1$), or $\{101,
\ldots , 200\}$ ($G2$). By construction, the hypothesis $H_{0, G1}$
should be accepted while $H_{0, G2}$ rejected. We consider decision
rules based on a significance level $\alpha = 0.05$ and 
reject $H_{0, G}$ if the event $E_G=\{\varrho_G \le 0.05\}$ occurs,
with $\varrho_G$ defined in (\ref{eq: pg}). To evaluate the
performance of these tests, we consider type I error and power, which can be represented by
$\hat{\mathbb P}\left(E_{G2}\right)$ and $\hat{\mathbb P}\left(E_{G1}\right)$, respectively, where again, $\hat{\mathbb P}$ denotes the empirical probability over 1000 simulations. Figure~\ref{fig: group_tests} visualizes the results.

\begin{figure}[h]
    \centering
    \begin{subfigure}[b]{0.47\textwidth}
        \includegraphics[width=\textwidth]{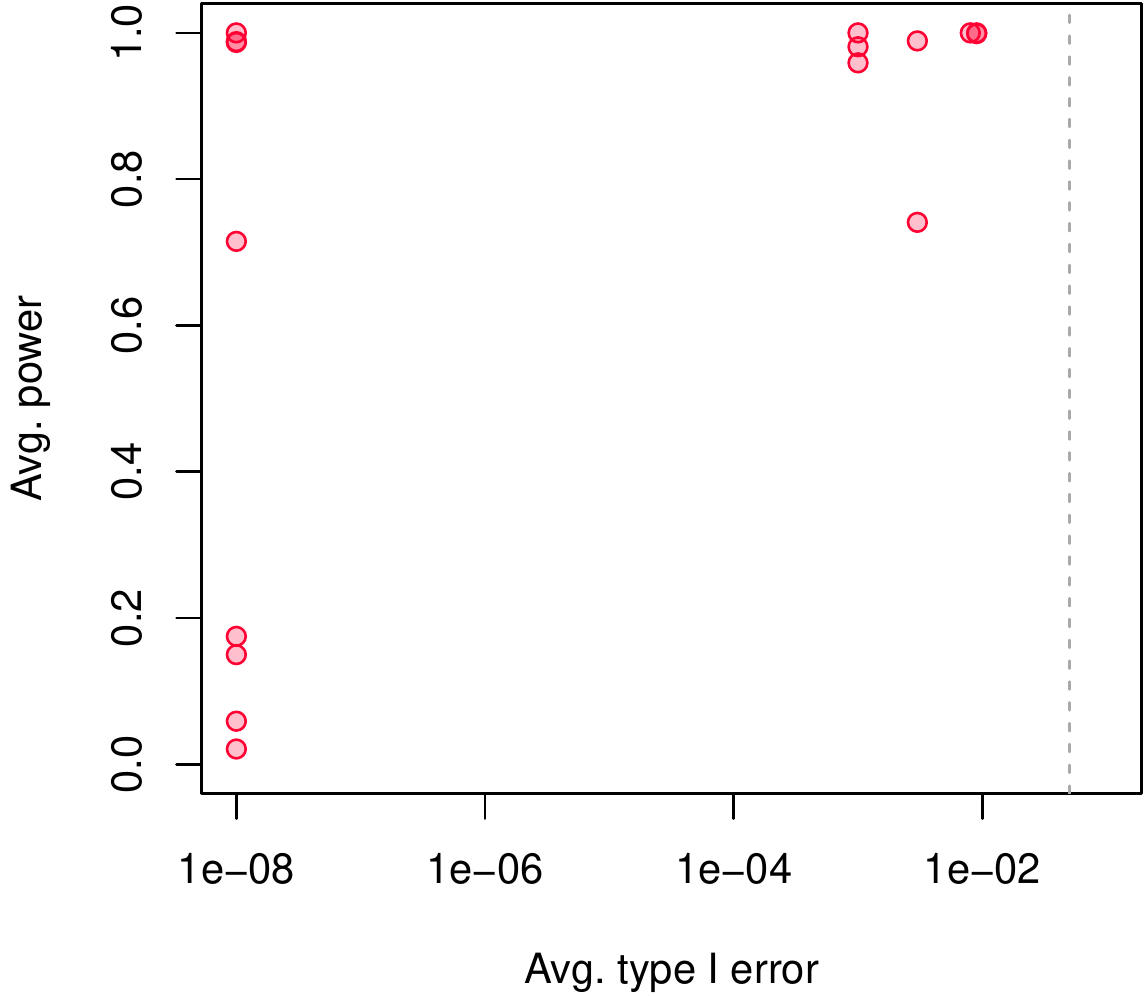}
        \caption{M1}
    \end{subfigure}
    ~ 
    \begin{subfigure}[b]{0.47\textwidth}
        \includegraphics[width=\textwidth]{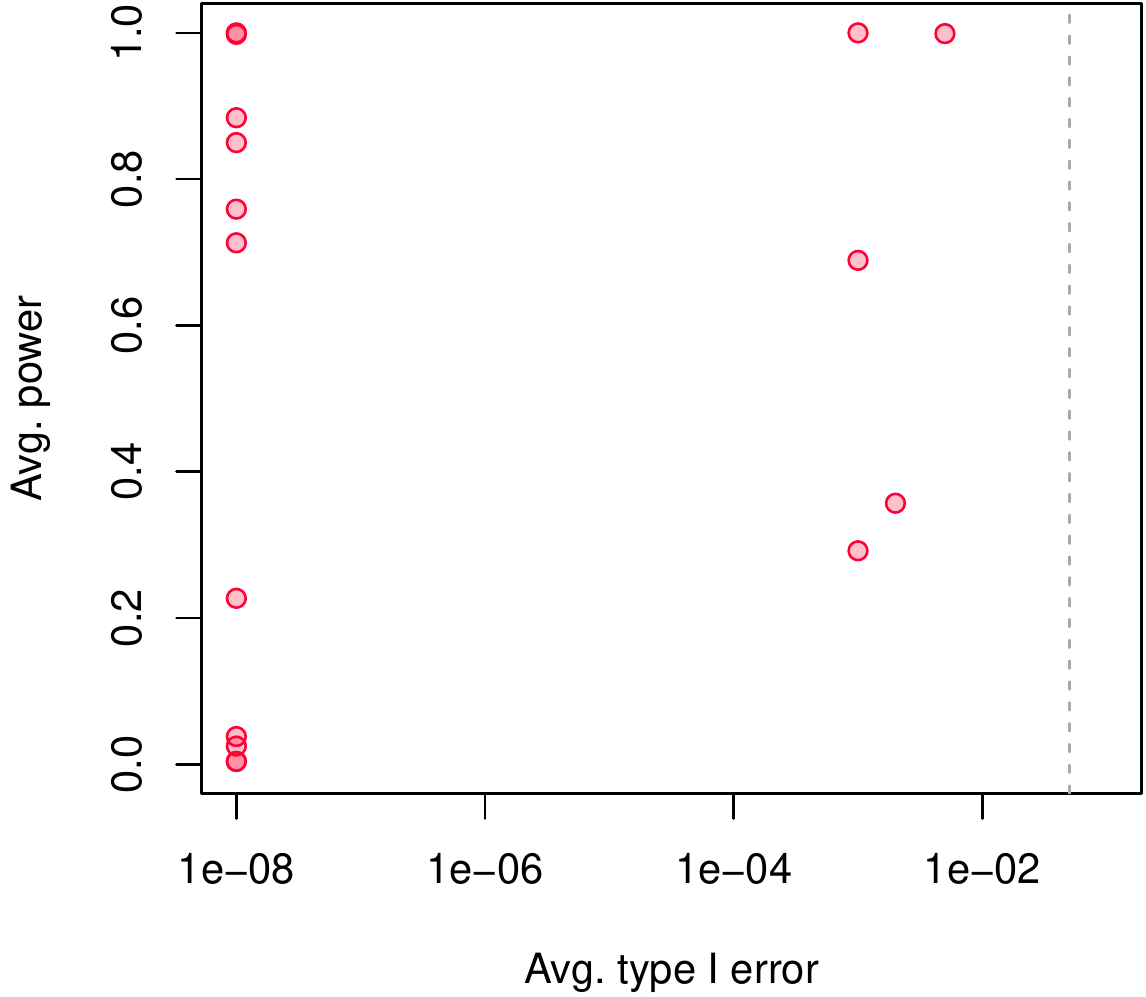}
        \caption{M2}
    \end{subfigure}
    \caption{Average power vs. average type I error for testing groups of coefficients under the two models for different combinations of $p$, $q$, $b$ and $d$.}\label{fig: group_tests}
\end{figure}

\subsection{Comparisons with existing methods}
In this section, we conduct a short numerical example to examine whether one could `naively' apply inferential procedures for high-dimensional linear models to obtain inference for parameters in mixed models.

Consider Model (M1) from Section \ref{sec: pval} in the instance of
$p = 300$ and $q = 1$. Let
$ \beta^* = [0.05, 2, 4, 3, 0.1, 0, \ldots, 0].  $ We compare our
method against
\begin{enumerate}
\item \textit{ridge}-based inference procedure of \citet{buhlmann2013}, which is an analogue of our method developed for high-dimensional linear models;
\item \textit{lasso}-based inference procedure of \citet{vdg2014}, which entailes de-sparsifying a lasso estimator.
\end{enumerate}
The differences are fairly evident when comparing confidence interval coverage. For any $\alpha \in (0,1)$, define $\mathbb Q_\alpha[W_j]$ as the $\alpha$-th quantile of the distribution of $W_j$. Under the conditions of Theorem \ref{thm: pvalue}, if the assumed model is correct, (\ref{eq: distribution}) suggests that confidence intervals of the form
\[
\left[\frac{\hat\beta_j^{\corr}}{\left(\mathbf P_{\mathbf X^T}\right)_{jj}} - \frac{\mathbb Q_{1-\alpha/2}\left[W_{j}\right] + C_j}{\left(\mathbf P_{\mathbf X^T}\right)_{jj}}, \frac{\hat\beta_j^{\corr}}{\left(\mathbf P_{\mathbf X^T}\right)_{jj}} + \frac{\mathbb Q_{1-\alpha/2}\left[W_{j}\right] + C_j}{\left(\mathbf P_{\mathbf X^T}\right)_{jj}} \right]
\]
should guarantee coverage of at least $(1-\alpha)$\%. Rather than setting $C_j$ according to Corollary \ref{cor: cj_choice}, we set them to be the same as the `$C_j$-analogues' from \cite{buhlmann2013}, to make the two methods comparable. Our choice of $C_j$ are larger than theirs, so if anything, this ad-hoc decision provides \cite{buhlmann2013}'s method an unfair advantage. In Figure \ref{fig: hdi_conf_int}, we examine 95\% confidence interval coverage for the three methods, based on the above modifications. 

\begin{figure}[t]
\centering
\includegraphics[scale = 0.8]{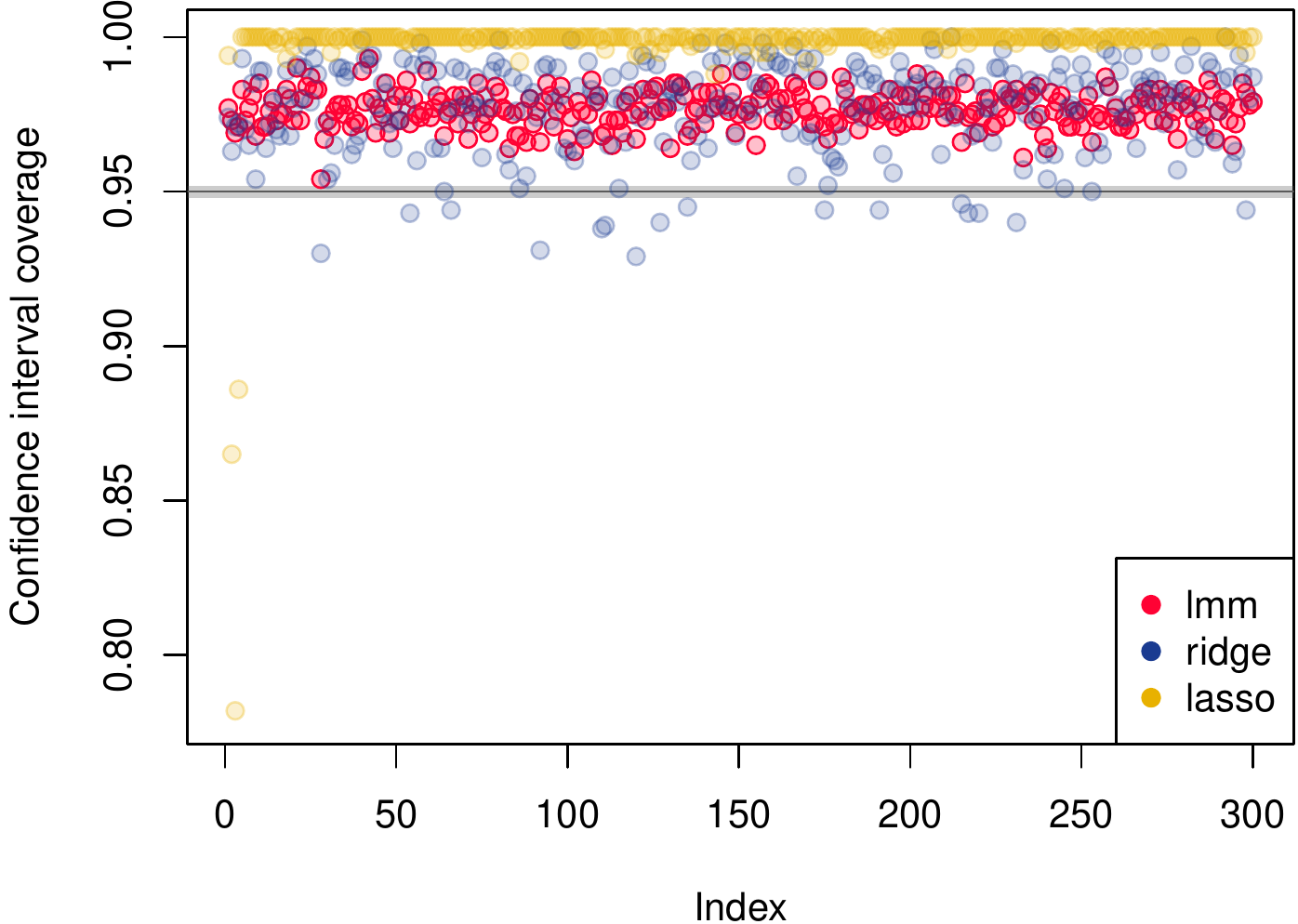}
\caption{Confidence interval coverage for $\beta^*_j$, $j = 1, \ldots, p$; target coverage is 95\% (with 1000 simulations, the standard deviation is $\sim$0.69\%). Here, \texttt{lmm} (\ref{firstcolour}) refers to our method; ridge (\ref{secondcolour}) to the method of  \citet{buhlmann2013}; and \texttt{lasso} (\ref{thirdcolour}) to the method of \citet{vdg2014}.}
\label{fig: hdi_conf_int}
\end{figure}

Overall, our method, which accounts for random effects, performs best
at attaining the target guaranteed coverage across all $\beta^*_j$'s,
compared to the methods proposed in \cite{buhlmann2013} and
\cite{vdg2014}. While \cite{buhlmann2013}'s method does come close,
coverage falls short at 16 indices: minimum coverage achieved was
92.9\% (with 1000 simulations, this is a statistically significance difference from
0.95). At initial glance it appears that the lasso-based method from
\cite{vdg2014} performs quite well; however, a closer examination of the results reveals otherwise. Specifically, the lasso-based method does very poorly over some of the active coefficients, as made evident in Table~\ref{tab:conf_int}.

\begin{table}
\caption{\label{tab:conf_int}Confidence interval coverage  for signals
  $\beta^*_j$, $j = 1, \ldots, 5$; target coverage is 95\%.}
\centering
{
  \begin{tabular}{l|c|c|c}
& Our method &  \cite{buhlmann2013} & \cite{vdg2014} \\ 
\hline
$\beta^*_1$ &  0.977 & 0.974 & 0.994  \\
$\beta^*_2$ &  0.973 & 0.963  &  0.865 \\
$\beta^*_3$ &  0.969 & 0.971  & 0.782  \\
$\beta^*_4$ &  0.971 & 0.972  & 0.886  \\
$\beta^*_5$ &  0.983 &  0.993 & 1.000 \\
\end{tabular}}
\end{table}

\section{An application to riboflavin production data}\label{sec: prac}

In this section, we apply our proposed methodology to data on
riboflavin (vitamin $B_2$) production by \textit{Bacillus
  subtilis}. The data is made publicly available by
\cite{buhlmann2014high}; the original data was provided by DSM
(Switzerland).  The dataset,  referenced as \textit{riboflavinGrouped}, has $M = 28$ specimens measured at two to six time points, resulting in $N = 111$ observations in total. For each specimen at each time point, we record a single real valued response variable, the log-transformed riboflavin production rate, as well as the expression levels of $p = 4088$ genes. We are interested in identifying which gene is significantly correlated with riboflavin production. 

To account for correlations induced by repeated measurements, a natural model to consider is the random intercept model, in which we assume that 
\begin{align}\label{eq: rint}
y_m = \mathbf X_m\beta^* + v_m + \epsilon_m,
\end{align}
with $v_m$, $m = 1, \ldots , M$ i.i.d.~with $v_m \sim N(0, \tau^{*2})$, and $\epsilon_m$, $m = 1, \ldots , M$, independent with $\epsilon_m \sim N(0, \sigma^{*2}\mathbf I_{n_m \times n_m})$, and generated independently of $v_1, \ldots , v_m$. Note that (\ref{eq: rint}) can be represented by (\ref{eq: mixed}) with the $\mathbf Z_m$'s taken to be column vectors of 1s of lengths $n_m$. Most of the theoretical results assume the $n_m$'s are equal, but it is straightforward to show the results hold so long as $n_m$ are on the same order of magnitude, as they are here. 

We apply our proposed framework and compute the marginal $p$-values for testing $\beta_j^* = 0$. Controlling the family-wise error rate (FWER) at 5\%, via a simple Bonferroni correction, we find a single significant gene in riboflavin production: \textit{YXLD-at}. This result matches previous findings by \cite{javanmard2014} and \cite{meinshausen2009split} using an homogeneous dataset with $N = 71$ samples provided by the same source (\textit{riboflavin} in \cite{buhlmann2014high}). Like us, \cite{meinshausen2009split} makes a single discovery, \textit{YXLD-at}, while \cite{javanmard2014} also labels \textit{YXLE-at} as significant. The method of \cite{buhlmann2013}, on the other hand, makes no discoveries.

\section{Discussion}\label{sec: discussion}

We presented a new framework for constructing asymptotically valid $p$-values and confidence intervals for the fixed effects in high-dimensional linear mixed effect models. It entails de-biasing a `naive' ridge estimator, whose asymptotic distribution we can approximate sufficiently well if the number of independent groups of observations $M$ scales at least with $d^2q\log p$.  
Simulation studies in high-dimensional suggest that our method provides good control of type-I error.  It also provides good results for a riboflavin dataset with group structure, where we confirmed results obtained in earlier work based on a homogeneous dataset from the same source \citep{javanmard2014, meinshausen2009split}.

Several extensions to our methodology would be of interest for future
work. First, our proposal for selecting the tuning parameter $\lambda_L$ relies on the assumption that $\tau^{*2}/\sigma^{*2}$ is not too small. 
Although it appears to work well in practice, one could also consider
an iterative scheme that repeatedly updates $\lambda_L$ based on the
resultant estimates of $\sigma^{*2}$ and $\tau^{*2}$: this can be
readily implemented in practice but may be difficult to validate
theoretically.  Second, here we required the number of random
effects $q$ to be quite small (treated as constant in the theory). This assumption can be relaxed by, e.g., taking $\mathbf \Psi^*$ to be a
general diagonal matrix, i.e., $\mathbf \Psi^* = \mbox{diag}(\tau^{*2}_1, \ldots , \tau^{*2}_q)$, and assuming that a small number of $\tau^{*2}_j$'s are nonzero, i.e., cardinality of $T \equiv \{j: \tau^{*2}_j \neq 0\}$ is small, less than $n$. Then, instead of screening for fixed effects in Step 1, we can screen for both fixed and random effects by incorporating a double penalization scheme as in \cite{wang2010}.  This way, in Step 2, both $|\hat S|$ and $|\hat T|$ are small, and we can apply Henderson's method III as before.

A few other details should also be discussed for completeness. First, multiple testing can be handled using the Westfall-Young procedure of  \cite{buhlmann2013}.  This multiple testing adjustment, which strongly controls the family-wise error rate, can directly be used in conjunction with our method for generating $p$-values for the individual hypothesis tests.
Second, the ridge-based framework of \cite{buhlmann2013}, which is a basis for our method, is known to \textit{not} have optimal power. \cite{buhlmann2013} shows that the detection rate may be larger than $N^{-1/2}$, whereas, under certain conditions, the detection limit for the  de-biased lasso approach of \cite{zhang2014} is in the $N^{-1/2}$ range. A possible extension of our work is to build a lasso-based inferential framework for high-dimensional linear mixed effect models. In fact, as suggested in the Introduction, our methods can be adapted to other high-dimensional estimators; and ridge is just an example. From \cite{vdg2014}, we can obtain asymptotically optimal inference for linear fixed effect models---i.e., for $\mathbf y = \mathbf X\beta^* + \epsilon$ with $N$ observations and $\epsilon_i$ i.i.d $N\left(0, \sigma^{*2}\right)$---by leveraging the fact that the Lasso estimator with non-negative penalty parameter $\lambda$, $\hat \beta(\lambda)$, can be rewritten as 
\[
\hat\beta(\lambda) - \beta^* + \lambda\hat{\mathbf \Theta} \hat \iota  =  \lambda\hat{\mathbf \Theta}\mathbf X^T\epsilon/N - \Delta/\sqrt{N}, \,\, \text{ where } \,\, 
\Delta := \sqrt{N}(\hat{\mathbf \Theta} \hat{\mathbf \Sigma} - \mathbf I_{p \times p})(\hat \beta(\lambda) - \beta^*)
\]
by inverting the KKT conditions, with $\hat \iota$ arising from the
subdifferential of $\|\beta\|_1$. Taking 
$\hat{\mathbf \Theta}$ to be a reasonably good approximation of an
inverse of $\hat{\mathbf \Sigma}$, the $\Delta$ term becomes
asymptotically negligible, and we can use the normality of $\epsilon$
to develop asymptotically valid tests and confidence intervals for
$\beta^*$.  (The scaled lasso furnishes a consistent estimator of
$\sigma^{*2}$.)  Extending this approach to the linear mixed-effect
setup (per Section \ref{sec: model}) requires meeting the challenge
that the $\epsilon_i$ are no longer i.i.d., which could be addressed using the methods of Section \ref{sec: consistent}.

\appendix

\section{Appendix}

\subsection{Proof of Results in Section~\ref{sec: method}}

\subsubsection{Proof of Lemma~\ref{lem:diagonal}}
It is straightforward to show that $\mathbf \Omega^*$ can be lower bounded as
\begin{equation*}
\mathbf \Omega^* \geq  c (\hat{\mathbf \Sigma} + \lambda
  \mathbf I_{p \times p})^{-1} \hat{\mathbf \Sigma} (\hat{\mathbf
    \Sigma} + \lambda \mathbf I_{p \times p})^{-1} \equiv \mathbf{\tilde
    \Omega}^*,
\end{equation*}
for some $c$ satisfying $0 < c <  \nu_{\min}\left(\mathbf V(\sigma^{*2},
\tau^{*2})\right)$. Since $\sigma^{*2}$ is positive, $\nu_{\min}\left(\mathbf V(\sigma^{*2}, \tau^{*2})\right) > 0$. Note that $\mathbf{\tilde \Omega}^*$ can 
alternatively be written as 
\begin{equation*}
\mathbf{\tilde \Omega}^* = \mathbf \Gamma ~\mbox{diag}\left(\frac{s_1^2}{(s_1^2 + \lambda)^2}, \ldots, \frac{s_N^2}{(s_N^2 + \lambda)^2} \right) \mathbf \Gamma^T,
\end{equation*}
which, in turn, implies that
\begin{equation*}
\tilde{\omega}^*_{\min} = \underset{j \in \{1, \dots , p \}}\min ~ \sum_{k = 1}^N \frac{s_k^2}{(s_k^2 + \lambda)^2}\mathbf \Gamma_{jk}^2,
\end{equation*}
and the claim follows.

\subsubsection{Proof of Proposition \ref{prop: ridge_theta}}

This was proven in \cite{shao2012} (see proof of their Theorem 1). Define $\mathbf \Gamma = \begin{bmatrix} \mathbf \Gamma' & (\mathbf \Gamma)_{\perp} \end{bmatrix}$; $\mathbf \Gamma'$ is orthogonal, i.e., $\mathbf \Gamma'^T\mathbf \Gamma' = \mathbf \Gamma' \mathbf \Gamma'^T = \mathbf I_{p \times p}$ . By definition (\ref{eq: ridge}), we have
\begin{align*}
\mathbb E[\hat \beta] - \theta^* 
&= \frac{1}{N}(\mathbf{\hat \Sigma} + \lambda \mathbf I_{p \times p})^{-1} \mathbf X^T\mathbf X \theta^* - \theta* \\
&= -(\lambda^{-1}N^{-1}\mathbf X^T \mathbf X + \mathbf I_{p \times p})^{-1} \theta^* \\
&= -\mathbf \Gamma'(\lambda^{-1}N^{-1}\mathbf \Gamma'^T\mathbf X^T\mathbf X \mathbf \Gamma' + \mathbf I_{p \times p})^{-1}\mathbf \Gamma'^T\mathbf \Gamma \mathbf \Gamma^T\theta^* \\
&= -\mathbf \Gamma(\lambda^{-1}N^{-1}\mathbf D^2 + \mathbf I_{R\times R})^{-1}\mathbf \Gamma^T \theta^*.
\end{align*}
Observing that the diagonal entries to $\mathbf D$ are positive, one obtains
\begin{equation}
(\lambda^{-1} N^{-1} \mathbf D^2 + \mathbf I_{R \times R})^{-1} \preceq \frac{\lambda^{-1}/\nu_{\min, +}(\hat{\mathbf \Sigma})}{1 + \lambda^{-1}/\nu_{\min, +}(\hat{\mathbf \Sigma})}\mathbf I_{R \times R},
\end{equation}
which, combined with the fact that $\mathbf \Gamma^T\mathbf \Gamma = \mathbf I_{R \times R}$, we obtain
\begin{equation*}
\underset{j \in \{1, \ldots, p\}}{\max}|\mathbb E[\hat \beta_j] - \theta_j^*| \le \lambda \|\theta^*\|_2  \nu_{\min, +}(\hat{\mathbf \Sigma})^{-1},
\end{equation*}
as desired. The bound on the variance follows directly from (\ref{eq: omega}). 

\subsection{Proof of Theorems \ref{thm: pvalue}, \ref{thm: main_var} and \ref{thm: ols}}

We first establish Theorem~\ref{thm: pvalue}, which follows directly from Proposition \ref{prop: ridge_theta}. 

\subsubsection{Proof of Theorem~\ref{thm: pvalue}}
It follows from Proposition~\ref{prop: ridge_theta} that
\begin{align*}
\max_j ~\kappa_j |\delta_j| &= \max_j ~\kappa_j \left|\mathbb E[\hat \beta_j] - \theta^*_j \right| 
\;\le\; \frac{\lambda \|\theta^*\|_2 \nu_{\min, +}\left(\hat{\mathbf \Sigma}\right)^{-1}}{N^{-1/2} \omega_{jj}^{*1/2}} 
\;\le\; \frac{\lambda \|\theta^*\|_2 \nu_{\min, +}\left(\hat{\mathbf \Sigma}\right)^{-1}}{N^{-1/2} \omega_{\min}^{*1/2}},
\end{align*}  
which, due to our choice of the ridge penalty parameter $\lambda > 0$
in (\ref{eq: lambda_ridge_choice}), is $o(1)$ as $N, p \to
\infty$. The claim now follows from (\ref{eq: dist}) and the assumption given by (\ref{eq: c_j}). 

Because there is an overlap in the lemmas used to prove Theorems \ref{thm: main_var} and \ref{thm: ols}, we present them together. Define $u^* = \|\mathbf X^T(y - \mathbf X\beta^*)\|_\infty/N$. 

\begin{lemma}\label{lem: conc}
  Let
\begin{equation}\label{eq: lasso_lambda}
\lambda_L = \frac{(\xi + 1)}{(\xi-1)}\sqrt{\frac{2(\sigma^{*2} + \tau^{*2} qn)(\log(p) - \log(\varepsilon/2))}{N}}.
\end{equation}
Under the model given by (\ref{eq: margin}), the event $u^* \le \lambda_L(\xi - 1)/(\xi + 1)$ occurs with probability greater than $1 - \varepsilon$. 
\end{lemma}

\begin{proof}
Define $u_j = x_j^T(y - \mathbf X\beta^*)/N$. Then $u^* = \max_j |u_j|$. Under model (\ref{eq: margin}), we observe that,
\begin{align*}
u_j \sim N(0, x_j^T\mathbf V(\sigma^{*2}, \tau^{*2}) x_j)
\end{align*}
It follows from the Gaussianity of $u_j$ (in fact, sub-Gaussianity would suffice) that
\begin{align}\label{eq: conc}
\mathbb P[|u_j| > \lambda_L(\xi - 1)/(\xi  + 1)] &\le 2e^{-\frac{\lambda_L^2(\xi - 1)^2/(\xi  + 1)^2}{2x_j^T\mathbf V(\sigma^{*2}, \tau^{*2}) x_j}} \le 2e^{-\frac{\lambda_L^2(\xi - 1)^2/(\xi  + 1)^2}{2N\nu_{\max}(\mathbf V(\sigma^{*2}, \tau^{*2}))}} \le \frac{\varepsilon}{p}. 
\end{align} 
The second inequality follows from the fact that the columns of $\mathbf X$ are standardized such that $\|x_j\|_2^2 = N  ~ \forall j$. For the third inequality, recall that the columns of $\mathbf Z$ are standardized such that $\|z_j\|_2^2 = n  ~ \forall j$, which implies that the largest eigenvalue of $\mathbf V(\sigma^{*2}, \tau^{*2})$, the true covariance of $y$, satisfies $\nu_{\max}(\mathbf V(\sigma^*, \tau^*)) \le \sigma^{*2} + \tau^{*2}qn$. The third inequality in (\ref{eq: conc}) is obtained by plugging in our choice of $\lambda_L$ (\ref{eq: lambdaL}). Employing a union bound, we then have
\begin{align*}
\mathbb P[u^* \leq \lambda_L(\xi - 1)/(\xi  + 1)] \ge 1-\sum_{j=1}^p  \mathbb P[|u_j| > \lambda_L(\xi - 1)/(\xi  + 1)] \geq 1 - \varepsilon,
\end{align*}
This is our desired result. 
\end{proof}

\subsubsection{Proof of Lemma \ref{lem: shat}}
We use arguments similar to those employed in the proof of Theorem 3 in \cite{ye2010}. Suppose that $u^* \le \lambda_L$. Define $h = \hat\beta^L - \beta^*$. The Karuhn-Kush-Tucker (KKT) optimality conditions for Lasso is given by 
\begin{equation*}
\begin{cases}
\frac{x_j^T(y - \mathbf X\hat \beta^L)}{N}  = \lambda_L\mbox{sign}(\hat \beta^L_j), & \hat \beta^L_j \neq 0 ,\\
\frac{x_j^T(y - \mathbf X\hat \beta^L)}{N}  \in \lambda_L[-1, +1], & \hat \beta^L_j = 0.
\end{cases}
\end{equation*}
With some rearrangement, the KKT conditions can be rewritten as
\begin{equation}\label{eq: kkt}
\frac{\mathbf X^T(y - \mathbf X\beta^*) - \hat{\mathbf \Sigma}h}{N} = \lambda_L\hat \iota
\end{equation}
with $\hat \iota \in \mathbb R^p$ and $\iota_j = \mbox{sign}(\hat \beta^L_j)$ if $j \in \hat S$ and $\iota_j \in [-1, +1]$ otherwise: the subdifferential which arises from $\|\beta\|_1$. Rearranging (\ref{eq: kkt}) and observing that $\mbox{sign}(\hat \beta^L_j) = \mbox{sign}(h_j)$ for $j \notin S$ yields 
\begin{align*}
h'^T\hat{\mathbf \Sigma}h \le (u^* + \lambda_L)\|h'_S\|_1 + (u^* - \lambda_L)\|h'_{S^c}\|_1 
\end{align*}
for all vectors $h'$ with $\mbox{sign}(h'_{S^c}) = \mbox{sign}(h_{S^c})$. If we take $h' = h$, one can see that $h \in \mathcal C(\xi, S)$:
\begin{align*}
& 0 \le h^T\hat{\mathbf \Sigma}h \le (u^* + \lambda_L)\|h'_S\|_1 + (u^* - \lambda_L)\|h'_{S^c}\|_1 \\
&\implies \|h'_{S^c}\|_1 \le \frac{(u^* + \lambda_L)}{(\lambda_L - u^*)} \|h'_S\|_1 \le \xi \|h'_S\|_1.
\end{align*}
On the other hand, setting $h'$ to be any vector so that for some $j \in S^c$, $h'_j = h_j$ and $0$ elsewhere gives
\begin{equation*}
h_j \hat{\mathbf \Sigma}_{j, \cdot} h \le (u^* - \lambda_L)|h_j| \le 0,
\end{equation*}
which implies that $h \in \mathcal C_-(\xi, S)$. The KKT conditions (\ref{eq: kkt}) also tell us that 
\begin{equation*}
\|\hat{\mathbf \Sigma}h\|_\infty \le u^* + \lambda_L,
\end{equation*}
which, when combined with the definition of $\zeta$ (\ref{eq: zeta}) yields 
\begin{equation*}
\|h\|_\infty \le \frac{u^* + \lambda_L}{\zeta},
\end{equation*}
which is the desired result.
In the event that $u^* \le \lambda_L(\xi - 1)/(\xi + 1)$, we have
\begin{equation*}
  \|h\|_\infty \le \frac{2\lambda_L}{\zeta (\xi + 1)}.
\end{equation*}

\subsubsection{Proof of Lemma \ref{lem: fp}}
The proof is adapted from that of Theorem 3 in \cite{sun2012}. By construction, $\hat \beta^L$ satisfies the KKT conditions from (\ref{eq: kkt}) 
which implies that 
\begin{align*}
\frac{|x_j^TX(\hat \beta^L - \beta^*)|}{N} &= \frac{|x_j^T(y - \mathbf X\hat\beta^L - \epsilon)|}{N} \\
&\geq \frac{|x_j^T(y - \mathbf X\hat\beta^L - \epsilon)|}{N} - \frac{|x_j^T\epsilon|}{N} \\
&\geq \lambda_L - u^*.
\end{align*}
For $\mathcal A \subseteq \hat S \backslash S$, such that $|\mathcal A | \le N'$, the previous inequality implies
\begin{align}
(\lambda_L - u^*)^2|\mathcal A| &\le \frac{\sum_{j \in A} |x_j^T\mathbf X(\hat \beta^L - \beta^*)|^2}{N^2} \nonumber\\
&=  \frac{\sum_{j \in A} (\mathbf Xh)^T x_jx_j^T(\mathbf X h)}{N^2} \le \frac{\kappa_+(N', S)\|\mathbf Xh \|_2^2}{N}. \label{eq: fp}
\end{align}
Going back to the KKT conditions in (\ref{eq: kkt}), we have, for arbitrary $h' \in \mathbb R^p$,
\[
\frac{(\mathbf X\hat\beta^L - \mathbf Xh')^T\mathbf Xh}{N} \le \lambda_L(\|h'\|_1 - \|\hat\beta^L\|_1) + u^*\|h' - \hat\beta^L\|_1,
\]
which, when combined with the fact that 
\[ 
2(\mathbf X\hat\beta^L - \mathbf Xh')^T\mathbf Xh = \|\mathbf X \hat \beta^L - \mathbf Xh'\|_2^2 + \|\mathbf Xh\|_2^2 - \|\mathbf X\beta^* - \mathbf Xh'\|_2^2
\]
gives the inequality 
\begin{align}
\frac{\|\mathbf Xh\|_2^2}{N} &\le \lambda_L(\|\beta^*\|_1 - \|\hat \beta^L\|_1) + u^*\|h\|_1 \nonumber\\
&\le (\lambda_L + u^*)\|h_S\|_1.
\end{align}
Thus, $h$ lies in the cone in (\ref{eq: cone}) in the event that $u^* \le \lambda_L(\xi - 1)/(\xi + 1)$ (by noting that the left-hand side is lower bounded by 0). By definition of $\kappa(\xi, S)$ from (\ref{eq: compat}),
\begin{align*}
\frac{\|\mathbf Xh\|_2^2}{N} \le \frac{(\lambda_L + u^*)^2d}{\kappa^2(\xi, S)},
\end{align*}
which, when combined with (\ref{eq: fp}) implies
\begin{equation*}
|\mathcal A| \le \frac{\kappa_+(N', S)\xi^2d}{\kappa^2(\xi, S)} < N',
\end{equation*}
by Assumption \ref{a: nprime}. 

\subsubsection{Proof of Theorem \ref{thm: main_var}}
Suppose that $u^* \le \lambda_L(\xi - 1)/(\xi + 1)$. Then by Lemmas \ref{lem: shat} and \ref{lem: fp} and the referenced assumptions within, we have 
\begin{align}
\| \hat\beta^L - \beta^*\|_\infty &\le \frac{2\xi\lambda_L}{(\xi+1)\zeta} \implies |\beta^*_j | \le  \frac{4\xi\lambda_L}{(\xi+1)\zeta} \quad \mbox{for all $j \in S \backslash \hat S$}, \\\
|\hat S \backslash S| &\le N' \implies |\hat S| \le N' + d \lesssim d.
\end{align}
Denote $\mathbf{\hat X} = \begin{bmatrix} \mathbf X_{\hat S} & \mathbf Z \end{bmatrix}$. Under candidate model (\ref{eq: wrong}), our variance component estimators (via Henderson's Method III) are given by
\begin{align*}
\hat \sigma^2 &= \frac{y^T\left(\mathbf I_{N \times N} - \mathbf P_{\hat{\mathbf X}}\right)y}{N - \mbox{rank}(\mathbf{\hat X})}, \\
\hat \tau^2 &= \frac{y^T\left( \mathbf P_{\hat{\mathbf X}} - \mathbf P_{\mathbf X_{\hat S}} \right)y - \hat \sigma^2[\mbox{rank}(\mathbf{\hat X}) - \mbox{rank}(\mathbf X_{\hat S})]}{\mbox{tr}\left[\mathbf Z^T\left(\mathbf I_{N \times N} - \mathbf P_{\mathbf X_{\hat S}}  \right)\mathbf Z \right]}.
\end{align*}
See (\ref{eq: sigma_hat}) and (\ref{eq: tau_hat}).

Consider the more interesting scenario where $|S \backslash \hat S| > 0$. If $S$ is contained within $\hat S$, then it is straightforward to show that variance component estimators are consistent (the true model is a sub-model of the proposed one). We first prove, under the given assumptions, that $|\hat \sigma^2 - \sigma^{*2}| = o_P(1)$. Write $S_O = S \backslash \hat S$, '$O$' for omitted. Then, 
\begin{align}
\hat \sigma^2 &= \frac{(\mathbf{\hat X}\beta^*_{\hat S} + \mathbf X_{S^O}\beta^*_{S^O} + \epsilon)^T(\mathbf I_{N \times N} - \mathbf P_{\hat{\mathbf X}})(\mathbf{\hat X}\beta_{\hat S} + \mathbf X_{S^O}\beta^*_{S^O} + \epsilon)}{N - \mbox{rank}(\mathbf{\hat X})} \nonumber \\
&= \frac{\beta_{S^O}^{*T}\mathbf X_{S^O}^T(\mathbf I_{N \times N} - \mathbf P_{\hat{\mathbf X}})\mathbf X_{S^O}\beta^*_{S^O}}{N - \mbox{rank}(\mathbf{\hat X})} + \frac{2\beta_{S^O}^{*T}\mathbf X_{S^O}^T(\mathbf I_{N \times N} - \mathbf P_{\hat{\mathbf X}})\epsilon}{N - \mbox{rank}(\mathbf{\hat X})} +  \frac{\epsilon^T(\mathbf I_{N \times N} - \mathbf P_{\hat{\mathbf X}})\epsilon}{N - \mbox{rank}(\mathbf{\hat X})}.  \label{eq: sigma_break}
\end{align}
We proceed to show that the three parts to (\ref{eq: sigma_break}) satisfy
\begin{align}
&\left|\frac{\epsilon^T(\mathbf I_{N \times N} - \mathbf P_{\hat{\mathbf X}})\epsilon}{N - \mbox{rank}(\mathbf{\hat X})} -\sigma^{*2}\right| = o_P(1), \label{eq: sigma_3} \\
& \left|\frac{2\beta_{S^O}^{*T}\mathbf X_{S^O}^T(\mathbf I_{N \times N} - \mathbf P_{\hat{\mathbf X}})\epsilon}{N - \mbox{rank}(\mathbf{\hat X})}\right| = o_P(1), \label{eq: sigma_2} \\
\mbox{and} \quad & \frac{\beta_{S^O}^{*T}\mathbf X_{S^O}^T(\mathbf I_{N \times N} - \mathbf P_{\hat{\mathbf X}})\mathbf X_{S^O}\beta^*_{S^O}}{N - \mbox{rank}(\mathbf{\hat X})} = o(1), \label{eq: sigma_1} 
\end{align}
which would suggest that $\hat \sigma^2$ is indeed consistent for
$\sigma^{*2}$.   Note that the term in (\ref{eq: sigma_1}) is $\mbox{Bias}(\hat \sigma^2) = \mathbb E[\hat \sigma^2] - \sigma^{*2}$.

\begin{enumerate}
\item \textit{Proving (\ref{eq: sigma_2})}: Let $\mathbf \Gamma_{\hat{\mathbf X}_{\perp}} \mathbf D_{\hat{\mathbf X}_{\perp}} \mathbf \Gamma_{\hat{\mathbf X}_{\perp}}^T$ represent the eigendecomposition of $\mathbf I_{N \times N} - \mathbf P_{\hat{\mathbf X}}$. We note that the latter is idempotent, implying that the diagonal matrix $\mathbf D_{\hat{\mathbf X}_{\perp}}$, which is of rank $N - \mbox{rank}(\hat{\mathbf X})$, has only $0$ and $1$s as its eigenvalues. It is straightforward to show that 
\begin{align*}
\mathbb E\left[\frac{2\beta_{S^O}^{*T}\mathbf X_{S^O}^T(\mathbf I_{N \times N} - \mathbf P_{\hat{\mathbf X}})\epsilon}{N - \mbox{rank}(\mathbf{\hat X})} \right] &= 0 \quad \mbox{and} \\
 \mbox{Var}\left[\frac{2\beta_{S^O}^{*T}\mathbf X_{S^O}^T(\mathbf I_{N \times N} - \mathbf P_{\hat{\mathbf X}})\epsilon}{N - \mbox{rank}(\mathbf{\hat X})} \right] &= \frac{4\sigma^{*2}\beta_{S^O}^{*T}\mathbf X_{S^O}^T\Gamma_{\hat{\mathbf X}_{\perp}} \mathbf D_{\hat{\mathbf X}_{\perp}} \mathbf \Gamma_{\hat{\mathbf X}_{\perp}}^T\mathbf X_{S^O}\beta_{S^O}^*}{[N - \mbox{rank}(\mathbf{\hat X})]^2} \\
&\le \frac{4\sigma^{*2}\| \Gamma_{\hat{\mathbf X}_{\perp}}^T\mathbf X_{S^O}\beta_{S^O}^* \|_\infty^2}{N - \mbox{rank}(\mathbf{\hat X})} \\
&\precsim \frac{4\sigma^{*2}\| \Gamma_{\hat{\mathbf X}_{\perp}}^T\mathbf X_{S^O}\|_\infty^2}{N - \mbox{rank}(\mathbf{\hat X})} \times \frac{d^2q\log p}{M} = o(1).
\end{align*}
Statement (\ref{eq: sigma_2}) then follows from Chebyshev's inequality. 
\item \textit{Proving (\ref{eq: sigma_3})}: Let $\chi^2_i(1)$, $i = 1,
  \ldots, N - \mbox{rank}(\mathbf{\hat X})$, be i.i.d~random variables
  following a $\chi^2$ distribution with 1 degree of freedom. Observe that,
\begin{align}
\frac{\epsilon^T[\mathbf I_{N \times N} -  \mathbf P_{\mathbf{\hat
  X}}]\epsilon}{N - \mbox{rank}(\mathbf{\hat X})} =_d \frac{\epsilon^T
  \mathbf D_{\hat{\mathbf X}_{\perp}}  \epsilon}{N -
  \mbox{rank}(\mathbf{\hat X})} &=_d \frac{\sigma^{*2} \sum_{i=1}^{N -
                                  \mbox{rank}(\mathbf{\hat X})}
                                  \chi^2_i(1)}{N -
                                  \mbox{rank}(\mathbf{\hat X})}
\end{align}
and
\begin{align}
\left|\frac{\sigma^{*2} \sum_{i=1}^{N - \mbox{rank}(\mathbf{\hat X})} \chi^2_i(1)}{N - \mbox{rank}(\mathbf{\hat X})} - \sigma^{*2}\right| &= o_P(1). \label{eq: sigma_slln}
 \end{align}
Here, (\ref{eq: sigma_slln})  follows from Lemma \ref{lem: fp} and the fact that $N \succsim N' + d + qM$ , which implies that $N - \mbox{rank}(\mathbf{\hat X}) \to \infty$ as $M \to \infty$. Applying the Strong Law of Large Numbers (SLLN) for i.i.d.~random variables, we arrive at (\ref{eq: sigma_3}). 
\item \textit{Proving (\ref{eq: sigma_1})}: We observe that
\begin{align*}
|\mbox{Bias}(\hat \sigma^2)|
&\le \| \mathbf \Gamma_{\mathbf{\hat X}_{\perp}}^T\mathbf X_{S_O} \|_\infty^2  \times \|\beta^*_{S_O} \|_\infty^2 \times d^2 \\
&\lesssim \frac{d^2 q \log(p)}{M} = o(1),
\end{align*}
and we have completed our proof that $\hat \sigma^2$ is consistent under the stated assumptions. 
\end{enumerate}
We now demonstrate that the same claim holds for $\hat \tau^2$. Expanding out $y$, we obtain, after some algebraic manipulation,
\begin{align}
\hat \tau^2 &= \frac{\beta_{S_O}^{*T}\mathbf X_{S_O}^T( \mathbf P_{\hat{\mathbf X}} - \mathbf P_{\mathbf X_{\hat S}}) \mathbf X_{S_O}\beta_{S_O}^*}{\mbox{tr}\left(\mathbf Z'^T\mathbf Z'\right)} + \frac{2\beta_{S_O}^{*T}\mathbf X_{S_O}^T( \mathbf P_{\hat{\mathbf X}} - \mathbf P_{\mathbf X_{\hat S}})\epsilon}{\mbox{tr}\left(\mathbf Z'^T\mathbf Z' \right)} + \frac{\epsilon^T( \mathbf P_{\hat{\mathbf X}} - \mathbf P_{\mathbf X_{\hat S}})\epsilon}{\mbox{tr}\left(\mathbf Z'^T\mathbf Z'\right)} \nonumber \\
&+ \frac{2\beta_{S_O}^{*T}\mathbf X_{S_O}^T(\mathbf I_{N \times N} - \mathbf P_{\mathbf X_{\hat S}})\mathbf Zv}{\mbox{tr}\left(\mathbf Z'^T\mathbf Z' \right)} + \frac{v^T\mathbf Z'^{T}\mathbf Z' v}{\mbox{tr}\left(\mathbf Z'^T\mathbf Z' \right)} +  \frac{2v^T\mathbf Z^T(\mathbf I_{N \times N} - \mathbf P_{\mathbf X_{\hat S}})\epsilon}{\mbox{tr}\left(\mathbf Z'^T\mathbf Z'\right)}\nonumber \\
& - \frac{\sigma^{*2}[\mbox{rank}(\hat{\mathbf X}) - \mbox{rank}(\mathbf X_{\hat S})]}{\mbox{tr}\left(\mathbf Z'^T\mathbf Z'\right)} - \frac{\mbox{Bias}(\hat \sigma^2)[\mbox{rank}(\hat{\mathbf X}) - \mbox{rank}(\mathbf X_{\hat S})]}{\mbox{tr}\left(\mathbf Z'^T\mathbf Z'\right)}, \label{eq: tau_parts}
\end{align}
where we have defined $\mathbf Z' = (\mathbf I_{N \times N} - \mathbf P_{\mathbf X_{\hat S}})\mathbf Z$.  We set out to prove that the terms in (\ref{eq: tau_parts}) satisfy 
\begin{align}
&\left|\frac{2\beta_{S_O}^{*T}\mathbf X_{S_O}^T( \mathbf P_{\hat{\mathbf X}} - \mathbf P_{\mathbf X_{\hat S}})\epsilon}{\mbox{tr}\left(\mathbf Z'^T\mathbf Z' \right)}\right| = o_P(1), \label{eq: tau_1} \\
&\left|\frac{\epsilon^T( \mathbf P_{\hat{\mathbf X}} - \mathbf P_{\mathbf X_{\hat S}})\epsilon}{\mbox{tr}\left(\mathbf Z'^T\mathbf Z' \right)}  - \frac{\sigma^{*2}[\mbox{rank}(\hat{\mathbf X}) - \mbox{rank}(\mathbf X_{\hat S})]}{\mbox{tr}\left[\mathbf Z'^T\mathbf Z'\right]}\right| = o_P(1), \label{eq: tau_2} \\
&\left|\frac{2\beta_{S_O}^{*T}\mathbf X_{S_O}^T(\mathbf I_{N \times N} - \mathbf P_{\mathbf X_{\hat S}})\mathbf Zv}{\mbox{tr}\left(\mathbf Z'^T\mathbf Z' \right)}\right| = o_P(1), \label{eq: tau_3}\\
&\left|\frac{v^T\mathbf Z^T(\mathbf I_{N \times N} - \mathbf P_{\mathbf X_{\hat S}})\mathbf Zv}{\mbox{tr}\left(\mathbf Z'^T\mathbf Z' \right)}  - \tau^{*2}\right| = o_P(1),  \label{eq: tau_4}\\
 &\left|\frac{2v^T\mathbf Z^T(\mathbf I_{N \times N} - \mathbf
   P_{\mathbf X_{\hat S}})\epsilon}{\mbox{tr}\left(\mathbf Z'^T\mathbf
   Z'\right)}\right| = o_P(1), \label{eq: tau_5}
\end{align}
and
\begin{align}
   & \frac{\beta_{S_O}^{*T}\mathbf X_{S_O}^T( \mathbf P_{\hat{\mathbf X}} - \mathbf P_{\mathbf X_{\hat S}}) \mathbf X_{S_O}\beta_{S_O}^*}{\mbox{tr}\left(\mathbf Z'^T\mathbf Z'\right)} - \frac{\mbox{Bias}(\hat \sigma^2)[\mbox{rank}(\hat{\mathbf X}) - \mbox{rank}(\mathbf X_{\hat S})]}{\mbox{tr}\left(\mathbf Z'^T\mathbf Z'\right)} = o(1). \label{eq: tau_bias}
\end{align}
Let $\mathbf Q_{\mathbf Z'} \mathbf D_{\mathbf Z'} \mathbf \Gamma_{\mathbf Z'}^T$ represent the singular value decomposition of $\mathbf Z'$, with $\mathbf Q_{\mathbf Z'}$ and $\mathbf \Gamma_{\mathbf Z'}$ of dimensions $N \times qM$ and $qM \times qM$, respectively. Additionally, write the eigendecompositions of $\mathbf P_{\hat{\mathbf X}} - \mathbf P_{\mathbf X_{\hat S}}$ and $\mathbf I_{N \times N} - \mathbf P_{\mathbf X_{\hat S}}$ as $\mathbf \Gamma_{\hat{\mathbf X} \perp \mathbf X_{\hat S}}  \mathbf D_{\hat{\mathbf X} \perp \mathbf X_{\hat S}} \mathbf \Gamma_{\hat{\mathbf X} \perp \mathbf X_{\hat S}}^T$ and $\mathbf \Gamma_{\mathbf X_{\hat S} \perp} \mathbf D_{\mathbf X_{\hat S} \perp} \mathbf \Gamma_{\mathbf X_{\hat S} \perp}^T$, respectively. Note that (\ref{eq: tau_4}) and (\ref{eq: tau_bias}) make up $\mbox{Bias}[\hat \tau^2] = \mathbb E[\hat \tau^2] - \tau^{*2}$. 

To avoid repetition, some of the proofs are presented in abbreviated form. 
\begin{enumerate}
\item \textit{Proving (\ref{eq: tau_1})}: Clearly,
\begin{align*}
\mathbb E \left[\frac{2\beta_{S_O}^{*T}\mathbf X_{S_O}^T( \mathbf P_{\hat{\mathbf X}} - \mathbf P_{\mathbf X_{\hat S}})\epsilon}{\mbox{tr}\left(\mathbf Z'^T\mathbf Z' \right)}\right] &= 0 \\ 
\mbox{Var} \left[\frac{2\beta_{S_O}^{*T}\mathbf X_{S_O}^T( \mathbf P_{\hat{\mathbf X}} - \mathbf P_{\mathbf X_{\hat S}})\epsilon}{\mbox{tr}\left(\mathbf Z'^T\mathbf Z' \right)}\right]  &= \frac{4\beta_{S_O}^{*T}\mathbf X_{S_O}^T( \mathbf P_{\hat{\mathbf X}} - \mathbf P_{\mathbf X_{\hat S}})\mathbf X_{S_O}\beta_{S_O}^{*}}{\mbox{tr}\left[\mathbf Z'^T\mathbf Z' \right]^2} \\
&= \frac{4\beta_{S_O}^{*T}\mathbf X_{S_O}^T\mathbf \Gamma_{\hat{\mathbf X} \perp \mathbf X_{\hat S}}  \mathbf D_{\hat{\mathbf X} \perp \mathbf X_{\hat S}} \mathbf \Gamma_{\hat{\mathbf X} \perp \mathbf X_{\hat S}}^T\mathbf X_{S_O}\beta_{S_O}^{*}}{\mbox{tr}\left[\mathbf Z'^T\mathbf Z' \right]^2} \\
&\le \frac{4 \times d^2 \times qM \times \|\mathbf \Gamma_{\hat{\mathbf X} \perp \mathbf X_{\hat S}} \mathbf X_{S_O} \|_\infty \times \|\beta_{S_O}^{*}\|_\infty^2}{\mbox{tr}\left[\mathbf Z'^T\mathbf Z' \right]^2} \\
&= o(1),
\end{align*}
following from (\ref{eq: z_dist}) in Assumption \ref{a: z} and Assumption \ref{a: something}. 
\item \textit{Proving (\ref{eq: tau_2})}: Orthogonality of $\mathbf \Gamma_{\hat{\mathbf X} \perp \mathbf X_{\hat S}}$ implies that  $\epsilon^T( \mathbf P_{\hat{\mathbf X}} - \mathbf P_{\mathbf X_{\hat S}})\epsilon =_d \epsilon^T \mathbf D_{\hat{\mathbf X} \perp \mathbf X_{\hat S}} \epsilon$, so
\begin{equation*}
\mathbb E\left[\frac{\epsilon^T( \mathbf P_{\hat{\mathbf X}} - \mathbf P_{\mathbf X_{\hat S}})\epsilon}{\mbox{tr}\left(\mathbf Z'^T\mathbf Z' \right)}\right] = \mathbb E\left[\frac{\epsilon^T \mathbf D_{\hat{\mathbf X} \perp \mathbf X_{\hat S}} \epsilon}{\mbox{tr}\left(\mathbf Z'^T\mathbf Z' \right)}\right] =  \frac{\sigma^{*2}[\mbox{rank}(\hat{\mathbf X}) - \mbox{rank}(\mathbf X_{\hat S})]}{\mbox{tr}\left[\mathbf Z'^T\mathbf Z'\right]}  
\end{equation*}
and, using properties of quadratic forms, we have
\begin{align*}
\mbox{Var}\left[\frac{\epsilon^T( \mathbf P_{\hat{\mathbf X}} - \mathbf P_{\mathbf X_{\hat S}})\epsilon}{\mbox{tr}\left(\mathbf Z'^T\mathbf Z' \right)}\right] &= \mbox{Var}\left[\frac{\epsilon^T \mathbf D_{\hat{\mathbf X} \perp \mathbf X_{\hat S}} \epsilon}{\mbox{tr}\left(\mathbf Z'^T\mathbf Z' \right)}\right] \\
&= \frac{2\sigma^{*4}\mbox{rank}(\mathbf D_{\hat{\mathbf X} \perp  \mathbf X_{\hat S}})}{\mbox{tr}\left(\mathbf Z'^T\mathbf Z' \right)^2} \\
&= \frac{2\sigma^{*4}\mbox{rank}(\mathbf Z')}{\mbox{tr}\left(\mathbf Z'^T\mathbf Z' \right)^2} \lesssim \frac{1}{M} = o(1),
\end{align*}
the latter relation following from (\ref{eq: z_sing}) in Assumption \ref{a: z}. This proves (\ref{eq: tau_2}). 
\item \textit{Proving (\ref{eq: tau_3})}: Proof is similar to that of (\ref{eq: tau_1}), as we note that 
\begin{align*}
\mathbb E\left[\frac{2\beta_{S_O}^{*T}\mathbf X_{S_O}^T(\mathbf I_{N \times N} - \mathbf P_{\mathbf X_{\hat S}})\mathbf Zv}{\mbox{tr}\left(\mathbf Z'^T\mathbf Z' \right)}  \right] &= 0, \\
\mbox{Var}\left[\frac{2\beta_{S_O}^{*T}\mathbf X_{S_O}^T(\mathbf I_{N \times N} - \mathbf P_{\mathbf X_{\hat S}})\mathbf Zv}{\mbox{tr}\left(\mathbf Z'^T\mathbf Z' \right)}  \right] &= \frac{4\tau^{*2}\beta_{S_O}^{*T}\mathbf X_{S_O}^T\mathbf Q_{\mathbf Z'} \mathbf D_{\mathbf Z'}^2\mathbf Q_{\mathbf Z'}^T \mathbf X_{S_O}\beta_{S_O}^{*} }{\mbox{tr}\left(\mathbf Z'^T\mathbf Z' \right)^2} \\ 
&\le \frac{4\tau^{*2} \mbox{tr}(\mathbf Z'^{T} \mathbf Z) \|\mathbf Q_{\mathbf Z'}^T\mathbf X_{S_O}\beta_{S_O}^{*} \|_\infty^2}{\mbox{tr}\left(\mathbf Z'^T\mathbf Z' \right)^2} \\
&= o(1),
\end{align*}
having applied Assumptions \ref{a: z} and \ref{a: something} here. 
\item \textit{Proving (\ref{eq: tau_4})}: As for (\ref{eq: tau_2}), using again properties of quadratic forms, we can show that 
\begin{align*}
\mathbb E\left[\frac{v^T\mathbf Z^T(\mathbf I_{N \times N} - \mathbf P_{\mathbf X_{\hat S}})\mathbf Zv}{\mbox{tr}\left(\mathbf Z'^T\mathbf Z' \right)}  \right] &= \tau^{*2} ,\\
\mbox{Var}\left[\frac{v^T\mathbf Z^T(\mathbf I_{N \times N} - \mathbf P_{\mathbf X_{\hat S}})\mathbf Zv}{\mbox{tr}\left(\mathbf Z'^T\mathbf Z' \right)}  \right] &= \frac{2\tau^{*4}\mbox{tr}(\mathbf D_{\mathbf Z'}^4)}{\mbox{tr}(\mathbf D_{\mathbf Z'}^2)^2} = o(1),
\end{align*}
the last relation the result of (\ref{eq: z_sing}) from Assumption \ref{a: z}.
\item \textit{Proving (\ref{eq: tau_5})}: We can rewrite
\begin{align*}
\frac{2v^T\mathbf Z'^T\epsilon}{\mbox{tr}\left(\mathbf Z'^T\mathbf Z'\right)} =  \frac{2v^T\mathbf \Gamma_{\mathbf Z'} \mathbf D_{\mathbf Z'} \mathbf Q_{\mathbf Z'}^T\epsilon}{\mbox{tr}\left(\mathbf Z'^T\mathbf Z'\right)} =_d \frac{\sigma^* \tau^*\sum_{i=1}^{qM} s_i B_i}{\sum_{i=1}^{qM} s_i^2}.
\end{align*}
where $B_i$, $i = 1, \ldots \mbox{rank}(\mathbf Z')$ are random variables formed as the product of two independent $N(0, 1)$ random variables. Then,
\begin{equation*}
\mbox{Var}\left(\sum_{i=1}^{qM} s_i B_i\right) = \sum_{i=1}^{qM} s_i^2,
\end{equation*}
which implies that
\begin{equation*}
\mbox{Var}\left(\frac{\sigma^* \tau^*\sum_{i=1}^{qM} s_i B_i}{\sum_{i=1}^{qM} s_i^2} \right)  = \frac{1}{\sum_{i=1}^{qM} s_i^2}
\end{equation*}
and since $\sum_{i=1}^{qM} s_i^2 \asymp qM$ by (\ref{eq: z_dist}) in Assumption \ref{a: z}, we have proven our claim (\ref{eq: tau_5}). 
\item \textit{Proving (\ref{eq: tau_bias})}: By the definition of $\mbox{Bias}[\hat \tau^2]$, it is clear that
\begin{align*}
|\mbox{Bias}(\hat \tau^2)| &\le  \frac{\beta_{S_O}^{*T}\mathbf X_{S_O}^T( \mathbf P_{\hat{\mathbf X}} - \mathbf P_{\mathbf X_{\hat S}}) \mathbf X_{S_O}\beta_{S_O}^*}{\mbox{tr}\left( \mathbf Z'^T\mathbf Z' \right)} + \frac{|\mbox{Bias}(\hat \sigma^2)|[\mbox{rank}(\hat{\mathbf X}) - \mbox{rank}(\mathbf X_{\hat S})]}{\mbox{tr}\left( \mathbf Z'^T\mathbf Z'\right)} \\
&\le \frac{\beta_{S_O}^{*T}\mathbf X_{S_O}^T\mathbf P_{\mathbf Z} \mathbf X_{S_O}\beta_{S_O}^*}{\mbox{tr}\left( \mathbf Z'^T\mathbf Z'\right)} + \frac{qM \times |\mbox{Bias}(\hat \sigma^2)|}{\mbox{tr}\left( \mathbf Z'^T\mathbf Z'\right)} \\
&\le \frac{qM \times d \times \|\mathbf \Gamma_{\mathbf Z} \mathbf X_{S^O}\|_\infty^2 \times \|\beta_{S_O}^*\|_\infty^2}{\mbox{tr}\left( \mathbf Z'^T\mathbf Z'\right)} + \frac{qM \times \mbox{Bias}(\hat \sigma^2)}{\mbox{tr}\left( \mathbf Z'^T\mathbf Z'\right)} \\
&\lesssim \frac{d^2 q \log(p)}{M} = o(1),
\end{align*}
where the second last relation follows from the proven claim that $\mbox{Bias}(\hat \sigma^2)$  is $o(1)$, (\ref{eq: z_dist}) in Assumption \ref{a: z} and Assumption \ref{a: something}.
\end{enumerate}
Since the event $u^* \le \lambda_L(\xi - 1)/(\xi + 1)$ occurs with probability greater than $1 - 1/p \to 1$ as $p \to \infty$,
\begin{align*}
| \hat \sigma^2 - \sigma^{*2}| &= o_P(1), \\ 
| \hat \tau^2 - \tau^{*2}| &= o_P(1) \\ 
\end{align*}
as claimed.  

\subsubsection{Proof of Theorem \ref{thm: ols}}
Suppose that $u^* \le \lambda_L(\xi - 1)/(\xi + 1)$, and write $S_O = S \backslash \hat S$. The OLS fit $\hat \beta^{\init}$ has a simple closed-form expression:
\begin{align*}
\hat \beta^{\init}_{\hat S} &= (\mathbf X_{\hat S}^T \mathbf X_{\hat S})^{-1}\mathbf X_{\hat S}^Ty \\
&= \beta^*_{\hat S} + (\mathbf X_{\hat S}^T \mathbf X_{\hat S})^{-1}\mathbf X_{\hat S}^T\mathbf X_{S_O}\beta^*_{S_O} +  (\mathbf X_{\hat S}^T \mathbf X_{\hat S})^{-1}\mathbf X_{\hat S}^T(y - \mathbf X\beta^*).
\end{align*}
and $\hat \beta^{\init}_{\hat S^c} = 0$. 
Thus, by triangle inequality, 
\begin{equation}\label{eq: beta_init_ineq}
\|\hat \beta^{\init}_{\hat S} - \beta^*_{\hat S} \|_1 \le \|(\mathbf X_{\hat S}^T \mathbf X_{\hat S})^{-1}\mathbf X_{\hat S}^T\mathbf X_{S_O}\beta_{S_O}^*\|_1 +  \|(\mathbf X_{\hat S}^T \mathbf X_{\hat S})^{-1}\mathbf X_{\hat S}^T(y - \mathbf X\beta^*)\|_1. 
\end{equation}
We proceed by first bounding the first term on the right-hand side of (\ref{eq: beta_init_ineq}). By Assumption \ref{a: ols},
\begin{align*}
\|(\mathbf X_{\hat S}^T \mathbf X_{\hat S})^{-1}\mathbf X_{\hat S}^T\|_2 &\le \sqrt{\nu_{\max}[(\mathbf X_{\hat S}^T \mathbf X_{\hat S})^{-1}\mathbf X_{\hat S}^T\mathbf X_{\hat S}(\mathbf X_{\hat S}^T \mathbf X_{\hat S})^{-1}]} \\
&= \sqrt{\nu_{\max}[(\mathbf X_{\hat S}^T \mathbf X_{\hat S})^{-1}]} \\
&= \sqrt{\frac{1}{N}\nu_{\max}\left[(\mathbf X_{\hat S}^T \mathbf X_{\hat S}/N)^{-1}\right]} \\
&= \sqrt{\frac{1}{N}/\nu_{\min}\left(\mathbf X_{\hat S}^T \mathbf X_{\hat S}/N\right)} \le \frac{1}{\sqrt{N \psi_-(N', S)}} \lesssim \frac{1}{\sqrt{N}} .
\end{align*}
This, in turn, implies that
\begin{align*}
\|(\mathbf X_{\hat S}^T \mathbf X_{\hat S})^{-1}\mathbf X_{\hat S}^T\mathbf X_{S_O}\beta_{S_O}^*  \|_1 &\le
\sqrt{N' + d}\|(\mathbf X_{\hat S}^T \mathbf X_{\hat S})^{-1}\mathbf X_{\hat S}^T\mathbf X_{S_O}\beta_{S_O}^*  \|_2 \\
&\le \sqrt{N' + d}\|(\mathbf X_{\hat S}^T \mathbf X_{\hat S})^{-1}\mathbf X_{\hat S}^T\|_2 \|\mathbf X_{S_O}\beta_{S_O}^* \|_2 \\
&\le \sqrt{N' + d}\|(\mathbf X_{\hat S}^T \mathbf X_{\hat S})^{-1}\mathbf X_{\hat S}^T\|_2 \|\mathbf X_{S_O}\|_F \|\beta_{S_O}^* \|_\infty \\
&\le  \sqrt{\frac{(N' + |S|)d}{{ \psi_-(N', S)}}}\frac{2\xi\lambda_L}{(\xi + 1)\zeta} \lesssim \sqrt{\frac{d^2q\log(p)}{M}} = o(1),
\end{align*}
where the last relation follows from Assumption \ref{a: nprime}. 

We proceed to bound the second component on the right-hand side of (\ref{eq: beta_init_ineq}). We observe that
\begin{align*}
\|(\mathbf X_{\hat S}^T \mathbf X_{\hat S})^{-1}\mathbf X_{\hat S}^T\epsilon\|_1 &\le \sqrt{N' + d} \|(\mathbf X_{\hat S}^T \mathbf X_{\hat S})^{-1}\mathbf X_{\hat S}^T(y - \mathbf X\beta^*)\|_2\\ 
 &=  \sqrt{N' + d}\left\|(\mathbf X_{\hat S}^T \mathbf X_{\hat S}/N)^{-1}\frac{\mathbf X_{\hat S}^T(y - \mathbf X\beta^*)}{N}\right\|_2 \\
&\le  \sqrt{N' + d}\vertiii{(\mathbf X_{\hat S}^T \mathbf X_{\hat S}/N)^{-1}}_2 \left\|\frac{\mathbf X_{\hat S}^T(y - \mathbf X\beta^*)}{N}\right\|_2 \\
&\le (N'  + d)\vertiii{(\mathbf X_{\hat S}^T \mathbf X_{\hat S}/N)^{-1}}_2 \left\|\frac{\mathbf X_{\hat S}^T(y - \mathbf X\beta^*)}{N}\right\|_\infty  \\
&\le  \frac{N'  + d}{\kappa_-(N', S)} \lambda_L \lesssim \sqrt{\frac{d^2q\log(p)}{M}} = o(1).
\end{align*}
From Lemma \ref{lem: shat},
\begin{equation*}
\|\hat \beta^{\init}_{S_O} - \beta^*_{S_O} \|_\infty  = \|\beta^*_{S_O}\|_\infty \lesssim \sqrt{\frac{q \log(p)}{M}},
\end{equation*}
which implies that 
\[\|\hat \beta^{\init}_{S_O} - \beta^*_{S_O} \|_1 \lesssim \sqrt{\frac{d^2q \log(p)}{M}} = o(1).\]
By Lemma \ref{lem: conc}, the event $u^* \le \lambda_L(\xi - 1)/(\xi + 1)$ occurs with probability exceeding $1 - 1/p$. Combined, we obtain the desired result.

\subsection{Empirical Evaluation of Assumption~4}\label{sec:A4emp}
To assess the stringency of Assumption~4 compared to the irrepresentability condition \citep{zhao2006}, we conduct a simulation study similar to \citet{zhao2006}, but customized to our mixed linear model setting. 

Consider the model $$y = X\beta^\ast + Z \nu + \varepsilon,$$ with $X \in \mathbb{R}^{nM \times p}$ and $Z \in \mathbb{R}^{nM \times qM}$, $\beta \in \mathbb{R}^p$, $\nu \in \mathbb{R}^p$ and $\varepsilon \in \mathbb{R}^{nM}$. 
We consider $q = 2$ random effects, $M = 25$ groups, and $n = 20$ samples within each group. Among the $p$ fixed effect covariates, we set $d$ to have nonzero coefficients and the rest to have zero coefficients. More specifically, we set $\beta^\ast = (\underbrace{1,\ldots,1}_{d},\underbrace{0,\ldots,0}_{p-d})^\top$.

To assess the stringency of the two assumptions, in each of $B=1000$ simulation replications, we randomly generate design matrices $X$ and $Z$ jointly as $[X, Z_u] \sim_{iid} N(0, \Sigma)$, where $Z_u$ is the un-blocked version of $Z$. The covariance matrix $\Sigma$ is generated from a $\mathrm{Wishart}(p+q, I_{p+q})$ distribution, and $X$ and $Z$ are scaled such that $\| x_j \|_2^2 = nM$ and $\| z_j \|_2^2 = n$. Following \citet{zhao2006}, we consider $p = 2^k$ for $k \in \{3,4,\ldots,8\}$, and set $d = t p/8$ for $t \in \{1,2,\ldots,8\}$. 

Let $A^\ast$ be the index of the true active set (hence $|A^\ast| = d$). Further, let $S \subset (A^\ast)^c$ with $|S| = \min(p-d, p)$ be a random subset of variables with zero coefficients. For $j \in A^\ast$, let $\widetilde{X} = [X_{A^\ast\backslash\{j\}}, X_S, Z]$ be the augmented design matrix, and denote $\widehat{\Sigma} = \widehat{X}^\top \widehat{X}/N$. 

With the above notations, the irrepresentability condition is satisfied if $T_\mathrm{IR} \equiv \max_{j \in A^\ast} T_{\mathrm{IR},j} < 1$, where 
\[
	T_{\mathrm{IR},j} = 
		\| 
			\widehat{\Sigma}_{\left( A^\ast \backslash\{j\} \right)^c, A^\ast \backslash\{j\}} \left( \widehat{\Sigma}_{\left( A^\ast \backslash\{j\} \right)^c, A^\ast \backslash\{j\}} \right)^{-1} \text{sign}\left(\beta_A^\ast \backslash\{j\}\right)
		\|_\infty.
\]
Assumption~4 involves a related quantity, $T_{4,j} = \| \Gamma_{\widetilde{X}} x_j \|_\infty$, where $\Gamma_{\widetilde{X}} D_{\widetilde{X}} \Gamma_{\widetilde{X}}^\top$ is the eigen-decomposition of $\widetilde{X}\left(\widetilde{X}\widetilde{X}^\top\right)^{-1}\widetilde{X}^\top$. 
However, this assumption is satisfied if $T_4 \equiv \max_{j \in A^\ast} T_{4,j} = O(1)$. Thus, to satisfy Assumption~4, we need a constant $C$, not dependent on $N$ and $p$, such that $T_4 < C$. While $C$ can be any large but fixed constant, in this simulation we consider a moderate value of $C = 5$. 

The proportion of simulated data sets, where the irrepresentability assumption and Assumption~4 are satisfied are shown in Tables~\ref{tbl:TIRprop} and \ref{tbl:T4prop}, respectively. As in \citet{zhao2006}, the results in Table~\ref{tbl:TIRprop} indicate that the irrepresentability assumption can be stringent, especially as the dimension $p$ and the number of nonzero coefficients $d$ increase. In contrast, the results in Table~\ref{tbl:T4prop} suggest that, for $C=5$, Assumption~4 is much more likely to hold. Moreover, the proportion of cases for which this assumption holds does not change with $p$ or $d$. While the appropriate choice of $C$ is generally unknown, the results in this simulation suggest that even with moderate values (in this case $C=5$) Assumption~4 is likely satisfied. 

\begin{table}[h]
\centering
\caption{Proportions of cases satisfying the irrepresentability condition \citep{zhao2006}.}\label{tbl:TIRprop}
\begin{tabular}{ l |c|c|c|c|c|c }
$T_{\text{IR}} < 1$ & $p=8$ & $p=16$ & $p=32$ & $p=64$ & $p=128$ & $p=256$\\
\hline
$d=p/8$ & 1 & 1 & 0.975 & 0.823 & 0.327 & 0.022 \\
$d=2p/8$ & 1 & 0.735 & 0.283 & 0.013 & 0 & 0 \\
$d=3p/8$ & 0.736 & 0.231 & 0.014 & 0 & 0 & 0 \\
$d=4p/8$ & 0.356 & 0.062 & 0.001 & 0 & 0 & 0 \\
$d=5p/8$ & 0.244 & 0.024 & 0 & 0 & 0 & 0 \\
$d=6p/8$ & 0.208 & 0.015 & 0 & 0 & 0 & 0 \\
$d=7p/8$ & 0.199 & 0.012 & 0 & 0 & 0 & 0 \\
\end{tabular}
\end{table}

\begin{table}[h]
\centering
\caption{Proportions of cases satisfying Assumption~4.}\label{tbl:T4prop}
\begin{tabular}{ l |c|c|c|c|c|c }
$T_{\text{4}} < 5$ & $p=8$ & $p=16$ & $p=32$ & $p=64$ & $p=128$ & $p=256$\\
\hline
$d=p/8$ & 0.998 & 0.999 & 0.997 & 0.996 & 0.994	& 0.99 \\
$d=2p/8$ & 1 & 1 & 0.996 & 0.995 & 0.99 & 0.982 \\
$d=3p/8$ & 0.998 & 0.999 & 0.994 & 0.992 & 0.985 & 0.973 \\
$d=4p/8$ & 0.999 & 0.998 & 0.991 & 0.995 & 0.984 & 0.954 \\
$d=5p/8$ & 0.996 & 0.995 & 0.993 & 0.989 & 0.976 & 0.938 \\
$d=6p/8$ & 1 & 0.997 & 0.993 & 0.987 & 0.967 & 0.93 \\
$d=7p/8$	& 0.996 & 0.995 & 0.995 & 0.988 & 0.956	& 0.927 \\ 
\end{tabular}
\end{table}

\scalebox{0}{%
\begin{tikzpicture}
    \begin{axis}[hide axis]
        \addplot [
        color=red,
        solid,
        line width=0.9pt,
        forget plot
        ]
        (0,0);\label{firstcolour}
            \end{axis}
\end{tikzpicture}%
}%

\scalebox{0}{%
\begin{tikzpicture}
    \begin{axis}[hide axis]
        \addplot [
        color=blue,
        solid,
        line width=0.9pt,
        forget plot
        ]
        (0,0);\label{secondcolour}
            \end{axis}
\end{tikzpicture}%
}%

\scalebox{0}{%
\begin{tikzpicture}
    \begin{axis}[hide axis]
        \addplot [
        color=mycolor2,
        solid,
        line width=0.9pt,
        forget plot
        ]
        (0,0);\label{thirdcolour}
            \end{axis}
\end{tikzpicture}%
}%

\bibliographystyle{abbrvnat} 
{\bibliography{prelim}}

\end{document}